%% file: main.tex
\documentclass[11pt]{article} 

\usepackage{settings}

\usepackage{soul}

\begin{document}

\begin{titlepage}
\begin{flushright}
UWThPh 2025-13
\end{flushright}

\vskip 2cm
\begin{center}
{\Large{\textbf{Classifying Isolated Symplectic Singularities \\[5pt] via 3d $\mathcal{N}=4$ Coulomb Branches}}}

\vspace{15mm}
{{\large 
Antoine Bourget${}^a$, Quentin Lamouret${}^a$

Sinan Moura Soys\"uren${}^{b,c}$, 
Marcus Sperling${}^b$}} 
\\[5mm]
\noindent {${}^a$\em  Institut de physique théorique, Université Paris-Saclay, CEA, CNRS, \\ 91191, Gif-sur-Yvette, France}\\
Email: {{\tt antoine.bourget@ipht.fr}}, \\ {{\tt quentin.lamouret@ipht.fr}}
\\[5mm]
\noindent {${}^b$\em University of Vienna, Faculty of Physics, Mathematical Physics Group,\\
Boltzmanngasse 5, 1090 Vienna, Austria,}\\
Email: {{\tt sinan.moura.soysueren@univie.ac.at}},\\ {\tt{marcus.sperling@univie.ac.at}}
\\[5mm]
\noindent {${}^c$\em University of Vienna, Vienna Doctoral School in Physics, \\ Boltzmanngasse 5, 1090 Vienna, Austria.}
\\[15mm]
\end{center}

\begin{abstract}
\begin{adjustwidth}{2cm}{2cm} 
Based on the Decay and Fission Conjecture, we provide a classification of unitary quivers whose 3d $\mathcal{N}=4$ Coulomb branches exhibit isolated singularities. This yields the complete list of isolated conical symplectic singularities that can arise in this way.
In the process, we identify three new families of stable quivers: two giving rise to previously unknown isolated symplectic singularities, and one offering a novel realization of a known family.
\end{adjustwidth}
\end{abstract}

\end{titlepage}

\tableofcontents

\section{Introduction}

\subsection*{Mathematical Motivations } 
Symplectic singularities were defined by Beauville in 1999 \cite{Beauville:2000} to capture the notion of singular symplectic varieties. The simplest examples of symplectic singularities are \emph{isolated}, meaning that the singular locus is a point, and \emph{conical}, meaning there is a $\mathbb{C}^{\ast}$ action compatible with the symplectic structure. We henceforth call them \emph{Isolated Conical Symplectic Singularities}, ICSSs for short.  
Well-known ICSSs include the Kleinian/Du Val surface singularities $\mathbb{C}^2\slash \Gamma_{ADE}$ and the closures of minimal nilpotent orbits $\overline{\mathcal{O}_{\min}} (\mathfrak{g})$ of semi-simple Lie algebras $\mathfrak{g}$. These encode all minimal transverse slices that can appear for nilpotent orbits of semi-simple Lie algebras \cite{kraft1980minimal,Kraft1982,fu2017generic}. 

Beauville \cite{Beauville:2000} raised the following question: What are more examples of ICSSs with trivial local fundamental group, beyond closures of minimal nilpotent orbits? This remained open for about 20 years, until recently.
In \cite{bellamy2023new} the authors identified a new such family $\mathcal{Y}(\ell)$ as singularities in the blowup of the quotient of $\mathbb{C}^{4}$ by the dihedral group of order 2$\ell$. Shortly after, \cite{namikawa2023remark} provided the construction of what we call $\overline{h}_{n,\sigma}$ singularities by using (toric) hyper-Kähler quotients.

Another possible source of examples are 3d $\mathcal{N}=4$ Coulomb branches \cite{Nakajima:2015txa,Braverman:2016wma}, which under certain assumptions have symplectic singularities \cite{Weekes:2020rgb,Bellamy:2023lqf}. Here, we focus on quiver gauge theories with unitary gauge groups, generalized by adding non-simply laced edges \cite{Nakajima:2019olw}. The isolated character of the singularity can be detected following the conjectural \emph{Decay and Fission algorithm} \cite{Bourget:2023dkj,Bourget:2024mgn}, which computes the stratification of such a Coulomb branch into partially ordered symplectic leaves (this poset is encoded in a ``Hasse diagram'', see Figure~\ref{fig:intro}). Based on this conjecture, we provide a \emph{classification of isolated conical symplectic singularities realized as 3d $\mathcal{N}{=}4$ Coulomb branches of quiver gauge theories with unitary gauge groups}. Remarkably, this encompasses all the previously known ICSSs --- except for $D$ and $E$ type surface singularities --- and adds two new infinite families to the list, see below. We also get a new Coulomb branch identification, see \eqref{eq:CBequality}.

Our classification is also of physical interest, as discussed in the next paragraph and illustrated in Figure~\ref{fig:intro}. This, however, can be skipped by readers interested only in the mathematical content.

\begin{figure}[t]
\begin{center}
\scalebox{.9}{\begin{tikzpicture}
\node[draw,align=center,color=purple,text width=3cm] (1) at (-4,3.2) {String Theory / \\ Brane System};
\node[draw,align=center,color=purple] (2) at (4,3.2) {SCFT};
\node[draw,align=center,text width=2cm,color=myGreen] (3) at (-4,0) {Magnetic Quiver};
\node[draw,align=center,text width=2cm] (4) at (4,0) {Conical Symplectic Singularity};
\node[draw,align=center,text width=2cm] (5) at (8,3) {Hilbert Series};
\node[draw,align=center,text width=2cm] (6) at (8,-3) {Hasse Diagram};
\draw[-{Latex[length=3.5mm]},purple] (2) -- node[above, sloped] {\footnotesize Higgs Branch} (4);
\draw[-{Latex[length=3.5mm]}] (4) -- node[above, sloped] {\footnotesize GrDim} (5);
\draw[-{Latex[length=3.5mm]},myGreen] (3) -- node[above, sloped] {\footnotesize Monopole Formula} (5);
\draw[-{Latex[length=3.5mm]}] (4) -- node[above, sloped] {\footnotesize Symp. Leaves} (6);
\draw[-{Latex[length=3.5mm]},myGreen] (3) -- node[above, sloped] {\footnotesize Decay and Fission} (6);
\draw[-{Latex[length=3.5mm]},purple] (1) -- node[above, sloped] {\footnotesize Intersections} (3);
\draw[-{Latex[length=3.5mm]},myGreen] (3) -- node[above, sloped] {\footnotesize $\mathcal{M}_C$} (4);
\draw[-{Latex[length=3.5mm]},purple] (1) -- node[above, sloped] {\footnotesize $\alpha ' \rightarrow 0$} (2);
\end{tikzpicture}}
\end{center}
\caption{\textbf{Black}: Symplectic singularities can be partially characterized by the (Hasse diagram of the) poset of symplectic leaves, or the Hilbert series of graded dimensions.  \textbf{\textcolor{myGreen}{Green}}: For symplectic singularities that are realized as 3d $\mathcal{N}=4$ Coulomb branches $\mathcal{M}_C$, these objects can be computed from physics-inspired tools such as the Monopole Formula and the Decay and Fission algorithm. \textbf{\textcolor{purple}{Purple}}: Symplectic singularities are realized in string theory through two mechanisms: Higgs branches of superconformal field theories (SCFTs), which extend beyond hyper-Kähler quotients, and magnetic quivers derived from brane intersections.
Remark: In general, it is not possible to go against the arrows.}
\label{fig:intro}
\end{figure}
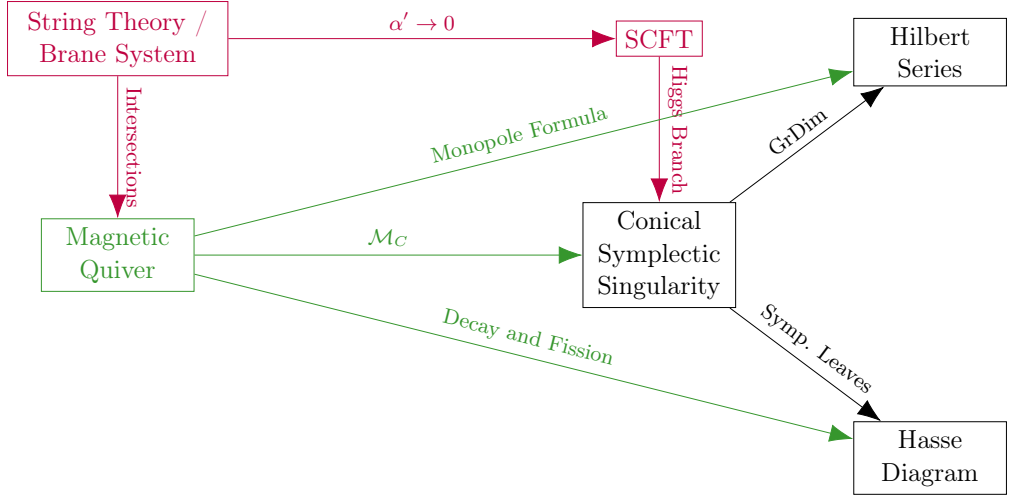

\subsection*{Physics Motivations} 
Superconformal quantum field theories (SCFTs) with 8 supercharges in spacetime dimensions $d=3,4,5,6$ typically exhibit a moduli space $\mathcal{M}$ of supersymmetric vacuum solutions. A distinguished branch $\mathcal{M}_H \subset \mathcal{M}$, called the Higgs branch, can be defined as the locus left invariant by all of the superconformal algebra, except for the $R$-symmetry factor $\mathfrak{su}(2)_R$. This is a conical symplectic singularity (CSS). Despite the difficulty of rigorously defining the quantum field theory framework (QFT), the Higgs branch is an object that is mathematically well-defined \cite{Hitchin:1986ea,Nakajima:1994nid,Antoniadis:1996ra}, and which can be used as a rich invariant for the SCFT --- indeed SCFTs can be defined using very diverse languages. The geometric properties of $\mathcal{M}_H$ are in one-to-one correspondence with features of the physical theory. For instance, isometries correspond to flavor symmetries and the finite stratification into partially ordered symplectic leaves corresponds to the Higgs mechanism \cite{Bourget:2019aer}. From that perspective, an ICSS corresponds to an elementary Higgs mechanism. More generally, ICSSs are viewed as elementary building blocks of Higgs branches, which motivates our effort to gather as many examples as possible.

In $d=3$ spacetime dimensions, another branch is also a CSS, the Coulomb branch $\mathcal{M}_C \subset \mathcal{M}$ \cite{Nakajima:2015txa,Braverman:2016pwk,Braverman:2016wma}. The coexistence of these two branches is the basis of 3d $\mathcal{N}=4$ mirror symmetry \cite{Intriligator:1996ex}, which exchanges the Higgs and Coulomb branches of two dual theories. Many insights into the Coulomb branch of Lagrangian theories have been developed in recent years (see \cite{Cremonesi:2013lqa,Bullimore:2015lsa} and subsequent works), building on the realization that monopole operators \cite{Borokhov:2002ib,Borokhov:2002cg} are a suitable starting point for the quantum behavior of these spaces. Following this, the 3d $\mathcal{N}=4$ Coulomb branch has been appreciated as a new construction method for symplectic singularities. In physics, this is particularly prominent in the magnetic quiver program (see \cite{Cabrera:2018jxt,Cabrera:2019izd,Bourget:2021siw} and subsequent works), which uses this new construction to study quantum Higgs branches of higher dimensional theories with 8 supercharges, as discussed in the previous paragraph.  This often relies on string theory which provides constructions of a vast class of SCFTs in dimension 3 to 6, using e.g. brane systems. In many cases, these allow one to derive magnetic quivers based on intersection numbers of branes in the magnetic phase, see for example \cite{Cabrera:2018jxt,Cabrera:2019izd}. Such magnetic quivers are, in the simplest cases (in particular in the absence of certain orientifold planes), so-called \emph{unitary quivers}. These are 3d $\mathcal{N}=4$ quiver gauge theories with gauge group $G=\prod_{i}\textrm{U}(n_{i})$, and generalizations thereof with so-called non-simply laced edges \cite{Cremonesi:2014xha}.

Here, we focus on symplectic singularities realized as the Coulomb branch of such unitary quivers. Within this class of theories, new isolated symplectic singularities have recently been found: 
\begin{itemize}
    \item The $\mathcal{Y}(\ell)$ singularities \cite{bellamy2023new} have been given a quiver realization in \cite{Bourget:2022tmw}.
    \item The $\overline{h}_{n,\sigma}$ singularities \cite{namikawa2023remark} have been realized as quivers in \cite{Bourget:2021siw,Bourget:2024asp}.
    \item Further singularities, called $\mathcal{J}_{2,3}$ and $\mathcal{J}_{3,3}$, have been found via unitary quivers \cite{Bourget:2022tmw}.
    \item Hyper-K\"ahler quotient singularities $h_{n,\delta,\sigma}$ by discrete cyclic groups have been realized as quivers in \cite{Bourget:2024asp} (building on earlier special cases of \cite{Bourget:2021siw}).
    \item Another (quaternionic) 4-dimensional singularity $gb_2$ has been conjectured in \cite{Bourget:2023dkj,Bourget:2024mgn}, which is extended here to a whole family $gb_n$ of new ICSSs.
\end{itemize}
The purpose of this work is to complete the classification of such isolated symplectic singularities, realized as Coulomb branches of 3d $\mathcal{N}=4$ quiver theories, by using the Decay and Fission algorithm \cite{Bourget:2023dkj,Bourget:2024mgn}.

\subsection*{Summary of Results and Organization of the Paper}
All ICSSs discussed above, arising from different constructions, remarkably show up in our classification, in a \emph{completely uniform language}, based on the Decay and Fission algorithm (Conjecture~\ref{conj:Decay-Fission}). Our first aim is to \emph{establish a framework} in which this conjecture can be precisely stated; this is the content of Section~\ref{sec:def}. Then our main result, stated and proven in Section~\ref{sec:proofs}, is: 
\vspace{4pt}
\begin{ftheo}\label{thm:final_result}
Assuming \textbf{Conjecture~\ref{conj:Decay-Fission}} holds, Table~\ref{tab:results} provides the complete list of unitary quivers (as defined in Definition~\ref{def:quiver}) whose Coulomb branches are ICSSs.
\end{ftheo}
\input{table_results}
Note in particular the addition of the new quiver families\footnote{Note that the quiver $gb_2$ was already discussed in the initial papers on the Decay and Fission algorithm \cite{Bourget:2023dkj,Bourget:2024mgn}, using the same method of derivation as in this work.} labeled $gb_n$, $gc_n$ and $gd_n$, which we claim to complete the list of unitary quivers whose Coulomb branch is an ICSS. 
As a first characterization, Table~\ref{tab:HWG} provides the isometry algebra as well as the Highest Weight Generating (HWG) function. Basic facts about Hilbert series and HWGs are reviewed in Appendix \ref{app:HS_and_HWG}. 
As with almost all known ICSSs\footnote{Recall, for $\overline{\mathcal{O}_{\text{min}}}\left(\mathfrak{g}\right)$, $\PL [ \mathrm{HWG} ]= \chi_{\mathrm{adj}} t^2 $. In contrast, the HWG for $\mathcal{J}_{2,3}$ and $\mathcal{J}_{3,3}$ do not have polynomial PL.}, the HWG has a polynomial plethystic logarithm (see Definition \ref{def:PL}), which is indicative of the simplicity of the moduli space. The HWG of $gc_n$ coincides with that of $\overline{h}_{2n+1,(3,1,\dots,1)}$, hinting at a new realization of that geometry, see Section \ref{sec:HWG_results}.  

Lastly, we relax the initial assumptions and conjecture the result to hold in a more general quiver setting, see Section~\ref{sec:pq_edge}.

\begin{table}[h]
    \centering
        \begin{tabular}{ccc} \toprule 
        ICSS & Symmetry & $\PL(\mathrm{HWG})$ \\ \midrule 
            $gb_n$ & $\mathfrak{so}_{2n+1}$ & $\mu_2 t^2 + (1 + \mu_1^2 + \mu_1^3) t^4 + \mu_1^3 t^6 -\mu_1^6 t^{12}  $ \\
            $gc_n$ & $\mathfrak{u}_1 \oplus \mathfrak{su}_{2n+1}$ & $(1+ \mu_1 \mu_{2n}) t^2 + (q \mu_1^3 + q^{-1} \mu_{2n}^3) t^4 - \mu_1^3 \mu_{2n}^3 t^8$ \\
            $gd_n$ & $\mathfrak{so}_{2n}$ & $\mu_2 t^2 + (1 + \mu_1^2 + \mu_1^3) t^4 + \mu_1^3 t^6 -\mu_1^6 t^{12}  $ \\ \bottomrule 
        \end{tabular}
    \caption{Plethystic logarithm of the HWG for the three new stable quiver families; two of which ($gb_n$ and $gd_n$) give rise to two new families of ICSSs. Note that the HWG for $gb_n$ and $gd_n$ are identical and independent of $n$. The $\mu_i$ are fugacities for the non-Abelian summand of the symmetry algebra and $q$ is the $\mathfrak{u}_1$ fugacity. }
    \label{tab:HWG}
\end{table}

\paragraph{Future Directions. }  We conjecture that the local fundamental group of the new families is trivial. It would be of great interest to prove this, in order to answer Beauville's question. More generally, one should aim at proving the isolated character of the symplectic singularities in our list directly from the Coulomb branch construction, and not using the Decay and Fission Conjecture. In turn, such a proof could constitute a first step in a proof of the conjecture itself.

\paragraph{Acknowledgments.}
We thank Paul Levy for sharing his insights regarding the new quiver families. We thank Julius Grimminger, Amihay Hanany, Daniel Juteau and Ben Webster for stimulating discussions. The work of SMS and MS is supported by the Austrian Science Fund (FWF), START project ``Phases of quantum field theories: symmetries and vacua'' STA 73-N [grant DOI: 10.55776/STA73]. SMS and MS also acknowledge support from the Faculty of Physics, University of Vienna.  SMS acknowledges the financial support by the Vienna Doctoral School in Physics (VDSP). The work of QL is supported by \'Ecole Normale Supérieure - PSL through a CDSN doctoral grant. 
MS gratefully acknowledges support from the Simons Center for Geometry and Physics, Stony Brook University, during the Workshop ``Symplectic Singularities, Supersymmetric QFT, and Geometric Representation Theory,'' at which the final stages of this work were performed.
\section{Definitions and Tools}
\label{sec:def}
\subsection{Good Quivers}

\begin{definition} \label{def:quiver}
A \textbf{quiver} $Q$ is a triple $(V,A,K)$ where $V$ is a finite set, $A$ a function $V\times V\to \mathbb{Z}$ and $K$ a function $V\to \mathbb Z_{>0}$, such that
\begin{assertions}
	\item for all $x \in V$, $A(x,x) = -2+2g_x$, for some $g_x \in \mathbb{Z}_{\geq 0}$. If $K(x) = 1$, then $A(x,x)=-2$.
    \item  there exists a function $L: V\to\mathbb Z_{>0}$, called a \emph{length function}, such that for every $x,y\in V$, $A(x,y) L(y)= A(y,x)L(x)$.
	\item for all $ x\neq y \in V$, if $A(x,y)$ and $A(y,x)$ are non-zero, then both are positive and one is a divisor of the other.
\end{assertions}
\end{definition}
\noindent Given a quiver $Q = (V,A,K)$ and an integer $N \geq 1$, we denote by $N \cdot Q$ the quiver $(V,A, N K)$. 

\begin{rmq}
\begin{compactitem}
\item Elements of $V$ are called \emph{vertices} of $Q$, $A$ the \emph{adjacency matrix}, and $K$ the \emph{weight function}. We say that $K(x)$ is the weight of $x \in V$.  When $V = \{ 1,\ldots,n\}$, we omit it and write $Q = (A,K)$ where $A$ is a matrix and $K$ a column vector.
Two vertices are neighbors if they are distinct and have a non-zero adjacency matrix coefficient. 
\item A Cartan matrix satisfying (ii) is known as symmetrizable Cartan matrix. Note here, that  property (iii) adds a restriction. We say that $Q$ is \emph{simply-laced} if its adjacency matrix is symmetric. Note that this allows for multiple edges between two vertices.  
\item In the gauge theory interpretation of quivers, the length functions are crucial in defining the gauge group, as reviewed in Appendix \ref{app:HS_and_HWG}. 
\item The integer $g_{x}$ is by definition the number of \emph{loops} for the vertex $x\in V$, and by definition a node of weight $1$ has no loops. 
\item Let $Q = (V,A,K)$ be a quiver. Its \emph{underlying graph} is the graph whose  set of vertices is $V$ and with a edge between $x,y\in V,x\neq y,$ when $A(x,y)\neq 0$. We say that a quiver is \emph{connected} (resp.\ \emph{cyclic, a tree, linear}, etc.) if its underlying graph is. The \emph{degree} of a vertex is its number of neighbors.
\item We say that two quivers are \emph{isomorphic} if there is a bijection between their sets of vertices preserving the adjacency matrices and the weight vectors.
\end{compactitem}
\end{rmq}

\begin{definition} \label{def:good}
    A connected quiver $Q$ is \textbf{quasi-good}\footnote{This notion weakens the goodness criterion of \cite{Gaiotto:2008ak}, and it turns out to be the one needed for the decay and fission conjecture. The underlying geometric translation is that quasi-good quiver might not have a conical Coulomb branch (they have only if the quiver is good), but they have a unique lowest-dimensional symplectic leaf, and therefore a unique maximal-dimensional transverse slice to a leaf. } if it is non-empty and one of the following holds true:
    \begin{assertions}
        \item Its Hilbert Series $\HS_Q(t)$ (see Definition~\ref{def:HS}) converges, is non-constant, and has in its series expansion no coefficient at order $t$.
        \item $Q$ is $N \cdot A_0^{(1)}$ for $N \geq 2$ or $N \cdot X_n^{(r)}$ for $N \geq 1$ and $X_n^{(r)}$ is one of the affine Dynkin quivers shown in Figure~\ref{Fig:Simply_Laced} and Table \ref{tab:dynkin}. 
    \end{assertions}
\end{definition}

A necessary condition \cite{Gaiotto:2008ak} for the ``quasi-good'' property reads
\begin{align}
\forall\, x\in V\, :\quad \sum_{y \in V} A(x,y)K(y) \geq 0 \; , \label{eq:good}
\end{align}
which is sufficient if $Q$ is simply-laced and not of type $N \cdot X_n^{(r)}$. 
\begin{rmq}
    \begin{itemize}
        \item The $N \cdot X_n^{(r)}$ quivers with $N>1$ are not stable, see Definition~\ref{def:stable}. Therefore, their Coulomb branches are not ICSSs. For $N \cdot X_n^{(r)}$ with $N=1$ both (i) and (ii) hold.
        \item If $\PL(\HS_Q(t)) = 2h\, t$ for some $h\in \mathbb{Z}_{>0}$, then $Q$ is said to be \emph{free}. If $\PL(\HS_Q(t)) = 2h \, t + \ldots$, then $Q$ is said to contain free parts. (See Definition~\ref{def:PL}, for the Plethystic Logarithm (PL).) 
   \end{itemize} 
\end{rmq}

\begin{definition}
Let $Q=  (V,A,K)$ be a quiver. A \textbf{subquiver} of $Q$ is a quiver $Q' = (V',A', K')$ where $V'\subseteq V$, $K'\leq K|_{V'}$ and $A'$ is defined by $A'(x,y) = -2$ if $x=y$ and $K(x) = 1$ and $A'(x,y) = A(x,y)$ otherwise.

If $V'$ is a subset of $V$, then the subquiver on $V'$ is $Q'  =(V',A', K|_{V'})$ with $A'$ defined as before.
\end{definition}

\begin{rmq}
A subquiver of a quasi-good quiver is not necessarily quasi-good. 
\end{rmq}

\subsection{Moduli Space of Vacua, Coulomb Branch, and Symplectic Leaves} 
A 3d $\Ncal=4$ gauge theory admits two maximal branches of the space of supersymmetric vacua: the Higgs and Coulomb branch. Here, the emphasis is placed on the Coulomb branch. Its mathematical definition was established in \cite{Nakajima:2015txa,Braverman:2016wma,Nakajima:2019olw}. It was later proven in \cite{Weekes:2020rgb,Bellamy:2023lqf} that Coulomb branches of all quiver gauge theories have symplectic singularities in the sense of \cite{Beauville:2000}. It then follows that 3d $\Ncal=4$ Coulomb branches admit a finite stratification into symplectic leaves \cite{Kaledin:2006}. As this is a finite partially ordered set, the stratification is naturally encoded in a Hasse diagram.

\begin{definition}
    Let $Q$ be a quasi-good quiver. We call $\mathcal{M}_C (Q)$ its \emph{Coulomb branch}, as defined in \cite{Nakajima:2019olw}. This is a conical symplectic singularity. 
\end{definition}

\subsection{The Decay and Fission Algorithm}

The Decay and Fission algorithm is reviewed here. It conjecturally provides a combinatorial construction of the Hasse diagram of $\mathcal{M}_C (Q)$.  We denote by $\leq$ the usual partial order on functions $V \rightarrow \mathbb{Z}$.

\begin{definition}
    Let $Q=(V,A,K)$ be a quiver. A \textbf{fission product} of $Q$ is a multiset $\multiset{Q_1,\ldots,Q_n}$ with $n\geq 0$, where $Q_i = (V_i,A_i,K_i)$ are quivers such that :
    \begin{assertions}
        \item for each $1\leq i\leq n$, $Q_i$ is a quasi-good connected subquiver of $Q$;
        \item $\sum_{i=1}^n K_i\leq K$, where we define each $K_i$ to vanish on $V\backslash V_i$
    \end{assertions}

    Let us call $\mathcal L(Q)$ the set of fission products of $Q$. 

    A \textbf{decay product} of $Q$ is a quiver $Q'$ such that $\multiset{ Q'} \in \mathcal L(Q)$. If $Q,Q'$ are two quivers such that $Q$ has a decay product isomorphic to $Q'$, we write $Q\rightsquigarrow Q'$.
\end{definition}

\begin{rmq}
    The empty fission product $\multiset{}$ is always in $\mathcal L(Q)$. Hence $Q$ is a decay product of itself if and only if it is quasi-good and connected.
\end{rmq}

The set of fission products can be equipped with a natural partial order:
\begin{definition}\label{def:partial_order}
    Let $Q$ be a quiver and $\multiset{Q_1,\ldots,Q_n},\multiset{Q'_1,\ldots,Q'_m}$ two fission products of $Q$. We write $\multiset{Q_1,\ldots,Q_n} \preccurlyeq \multiset{Q'_1,\ldots,Q'_m}$ if there exists a partition $\{ 1,\ldots,n\} = \bigsqcup_{j=1}^m I_j$, with the $I_j$ possibly empty, such that for every $1\leq j \leq m$, $\multiset{Q_i:i\in I_j}\in\mathcal L(Q'_j)$. 
\end{definition}

This leads to the following conjecture:
\begin{conjecture}[Decay and Fission algorithm \cite{Bourget:2023dkj,Bourget:2024mgn}]
\label{conj:Decay-Fission}
There exists a 1-to-1 correspondence between the poset of symplectic leaves of the Coulomb branch of a quasi-good 3d $\mathcal{N}=4$ quiver theory $Q$ and the poset $(\mathcal{L} (Q), \succcurlyeq)$ of decay and fission products of the quasi-good quiver $Q$. 
\end{conjecture}
\begin{rmq}
    The Decay and Fission algorithm also allows us to determine the minimal transitions, or, transverse slices in the form of a quiver. Using this, the underlying geometry can be determined.
\end{rmq}

\begin{definition}\label{def:stable}
 A quiver $Q$ is \textbf{stable} if it is quasi-good, non-empty, and its only decay product is itself. 
\end{definition}
\begin{rmq}
This is equivalent to $\mathcal L(Q) = \{\multiset{},\multiset{Q}\}$. Geometrically, it means that the Coulomb branch of $Q$ only has two leaves: the singular point and the regular part of the Coulomb branch. For a conical symplectic singularity, this is equivalent to the singularity being isolated.  
\end{rmq}

\section{Results and Proofs}
\label{sec:proofs}
In this section, Theorem~\ref{thm:final_result} is proven first for simply-laced and then for non-simply-laced quivers assuming Conjecture~\ref{conj:Decay-Fission}.
Before proceeding, we recall two facts:

\vspace{4pt}
\begin{theo}[Abelian case\cite{Bourget:2024asp}] \label{thm:list_stable_abelian_quivers}
The stable Abelian quivers\footnote{A quiver is said to be Abelian if $K(x)=1$ for all $x \in V$. } are given in Table~\ref{tab:results}.
\end{theo}
Moreover, any quiver that contains either $(i)$ two nodes connected with a simply-laced edge of multiplicity $N\geq 2$ or $(ii)$ a vertex $x\in V$ with $g_{x}\geq1$ cannot be a stable quiver, as it admits a decay product isomorphic to the quivers with geometry $A_{N-1}$ or $D_{g_x +1}$ (if $g_x > 1$) or $A_1$ (if $g_x = 1$), see Table~\ref{tab:results}. Therefore, any such quiver can be omitted in the following discussion.

\subsection{Simply-laced Quivers}
\label{sec:simply_laced}
To begin with, we consider quivers with simply-laced edges.
\begin{proposition}\label{prop:list_stable_simply_laced_quivers}
The quivers in Figure~\ref{Fig:Simply_Laced} are all quasi-good and stable. 
\end{proposition}
\input{fig_simply_laced}
\begin{proof}
    The two-vertex quiver has the $A$-type Kleinian surface singularity $\mathbb{C}^2\slash \mathbb{Z}_{N}$ as Coulomb branch. The untwisted affine Dynkin quivers $A_n^{(1)}$, $D_{n}^{(1)}$, $E_{6,7,8}^{(1)}$ (see Figure~\ref{Fig:Simply_Laced}) have Coulomb branches given by the closure of the  minimal nilpotent orbit $\overline{\mathcal{O}}_{\min}^{\mathfrak{g}}$, all of which are isolated symplectic singularities.

    Alternatively, it can be checked that these quivers contain no non trivial subquiver satisfying the necessary condition \eqref{eq:good}. 
\end{proof}
The main result in this section is the following theorem:

\vspace{4pt}
\begin{theo}\label{thm:simply_laced_decay}
Let $Q$ be a non-empty quasi-good simply-laced quiver. Then one of the quivers listed in Figure~\ref{Fig:Simply_Laced}  is a decay product of $Q$.
\end{theo}
\begin{proof}
Let $Q$ be a non-empty quasi-good simply-laced quiver. Each connected component of $Q$ is a decay product, therefore we only have to prove the theorem for connected quivers. If $Q$ contains a cycle, then this cycle, with all weights set to $1$, is a decay product isomorphic to $A^{(1)}_{n}$ for some $n$. So we only have to prove the theorem for acyclic quivers.

As $Q$ is finite and acyclic, it must have a node with exactly one neighbor. Then \eqref{eq:good} implies that the neighbor must have weight greater than or equal to $2$. Therefore, $Q$ is non-Abelian.

We can construct a decay product of $Q$ in the following way: choose a connected component in the subgraph of non-Abelian nodes. Add all the nodes with weight $1$ that are direct neighbors of these. If needed, we can remove a weight one vertex to ensure that the subquiver does not correspond to an over-extended $N \cdot X_n^{(r)}$ quiver\footnote{An over-extended $N\cdot X_n^{(r)}$ quiver includes an extra weight one vertex at the affine vertex of the Dynkin diagram.}. The resulting subquiver is quasi-good and, therefore, a decay product. It has the property that all $U(1)$ nodes have exactly $1$ neighbor. Hence proving the theorem for this decay product, implies the result for the original quiver. Therefore, \textit{we can assume that all nodes of weight $1$ have only one neighbor}.

If $Q$ has a node with four neighbors or more, then $Q$ decays to $D^{(1)}_4$. If $Q$ has two nodes with three neighbors, then it decays to $D^{(1)}_n$ for some $n \geq 5$. If $Q$ has no vertices with three or more neighbors, then it is linear. Let $a_1,\ldots,a_n$ be the weights of its nodes. Then, setting $a_0 = a_{n+1} = 0$, the ``good'' constraint \eqref{eq:good} reads:
\begin{equation}
    \forall k \in\{1,\ldots,n\}: \quad  2a_k \leq a_{k-1} + a_{k+1} \,,
\end{equation}
which is a convexity inequality. It implies :
\begin{equation}\label{eq:convexity_1}
    \forall k \in\{1,\ldots,n\}: \quad   a_k \leq \frac{n+1-k}{n+1} a_0 + \frac{k}{n+1}a_{n+1} = 0 \,.
\end{equation}
Therefore, we can assume that $Q$ has exactly one node with three neighbors, with all other nodes having one or two: $Q$ is of the form
\begin{align}
    \raisebox{-.5\height}{
\begin{tikzpicture}
    \tikzset{node distance = 0.75cm};
    \node (gg1) [gauge,label=below:{$a_2$}] {};
    \node (gg2) [gauge,right of=gg1,label=below:{$a_1$}] {};
    \node (gg3) [gauge,right of=gg2,label=below:{$m$}] {};
    \node (gg4) [gauge,right of=gg3,label=below:{$b_1$}] {};
    \node (gg5) [gauge,right of=gg4,label=below:{$b_2$}] {};
    \node (tt1) [gauge,above of=gg3,label=right:{$c_1$}] {};
    \node (tt2) [above of=tt1] {$\vdots$};
    \node (lhs) [left of=gg1] {$\cdots$};
    \node (rhs) [right of=gg5] {$\cdots$};
    \draw (lhs)--(gg1) (gg1)--(gg2) (gg2)--(gg3) (gg3)--(gg4) (gg4)--(gg5) (gg5)--(rhs) (gg3)--(tt1) (tt1)--(tt2);
\end{tikzpicture}
}
\label{eq:quiver_type_E}
\end{align}
Let $a_0=b_0=c_0 =m$ be the weight of the trivalent node, and $(a_1,a_2,\ldots)$, $(b_1,b_2,\ldots)$, $(c_1,c_2,\ldots)$ the sequences of weights along the three legs. We have $m\geq 2$ and $a_1,b_1,c_1 \geq 1$. As before, these sequences are convex and vanish eventually. A convexity inequality like \eqref{eq:convexity_1} implies that $a_1,b_1$ and $c_1$ are strictly smaller than $m$, and more generally that the sequences are strictly decreasing until they reach zero. The ``good'' constraint \eqref{eq:good} at the trivalent node reads:
\begin{equation}\label{eq:proof_E_balance1}
    2m\leq a_1 + b_1+c_1 \, . 
\end{equation}
Together with $a_1,b_1,c_1\leq m-1$, it implies $m\geq 3$.

If $a_2,b_2$ and $c_2$ are greater than or equal to $1$, then, as the sequences are strictly decreasing, $a_1,b_1$ and $c_1$ are greater than or equal to $2$ and $m \geq 3$. In this case, $Q$ decays to $E^{(1)}_6$. Therefore we can assume, without loss of generality, that $c_2= 0$. 

If $a_3,b_3 \geq 1$, then the sequences being strictly decreasing implies $m \geq 4$, and $Q$ decays to $E^{(1)}_7$.

Therefore we can assume that $b_3 =c_2= 0$. We then have $a_1\leq m-1$, $b_1 \leq \frac{2m}3$ and $c_1 \leq \frac{m}{2}$. Inserting two or three of these into \eqref{eq:proof_E_balance1}, we find $2m \leq \frac{13}{6}m -1$, hence $m\geq 6$, and
\begin{equation}
    a_1  \geq \frac{5m}{6} \, , \qquad b_1  \geq \frac{m}{2} + 1 \, , \qquad c_1   \geq \frac{m}{3}+1 \geq 3 \, . 
    \label{eq:proof_E_balance2}
\end{equation}

Convexity also implies that $ \forall k\geq 0, a_k\geq (1-k)m +ka_1$ and similar equations for the other two sequences. Using these in combination with \eqref{eq:proof_E_balance2}, we get, for all $k\geq 0$
\begin{equation}
  a_k  \geq \frac{m}{6}(6-k)  \, , \qquad    b_k  \geq \frac{m}{2}(2-k) +k  \, . 
\end{equation}
Finally, we find that for $k\in \{0,1,\ldots,6\}$, $a_k\geq 6-k$ and for $k\in\{0,1,2\}$, $b_k \geq 6-2k$. Therefore $Q$ decays to $E^{(1)}_{8}$. This concludes the proof.
\end{proof}

\begin{corollaire}
The quivers of Proposition~\ref{prop:list_stable_simply_laced_quivers} are exactly all the stable simply-laced quivers.
\end{corollaire}
\begin{proof}
None of those quivers are decay products of the others.  
\end{proof}

\input{Dynkin_Table}

\subsection{From Quiver Classification to Geometry Classification}\label{subsec:trivial_edges}

It is well-known that two non-isomorphic quivers $Q$ and $Q'$ can have the same Coulomb branch, $\mathcal{M}_C (Q) = \mathcal{M}_C (Q')$. So even though $Q$ and $Q'$ are distinct as combinatorial objects, they represent the same geometry. In this section, we provide a class of instances of this phenomenon which is crucial for our classification, see Table~\ref{tab:dynkin}. This amounts to enlarging slightly what we call a Dynkin graph, so as to encompass all quivers whose Coulomb branch is a minimal nilpotent orbit closure.

\begin{proposition}
    Let $Q = (V,A,K)$ be a quasi-good, stable, non-Abelian quiver. Assume $Q$ has an edge $E$ between a vertex $x$ of weight $K(x)=1$ and a vertex $y$ of weight $K(y) \geq 2$, with $A(x,y)=A(y,x)=1$. Let $Q' = (V,A',K)$ where $A=A'$ except for $A'(x,y) = \ell \geq 1$, and $Q'_{\rm{ab}} = (V,A',\bf{1})$ the abelian quiver obtained by setting the weights to 1. If $Q'_{\rm{ab}}$ is free then $\mathcal{M}_C (Q) = \mathcal{M}_C (Q')$. 
\end{proposition}

\begin{proof}
The quiver $Q$ (and likewise for $Q'$) can be non-simply-laced, and therefore may not define a morphism $ \prod_x U(K(x))\to \prod_{\langle x;y\rangle} U(K(x) K(y))$, where the first product runs over the vertices and the second product runs over the edges.  
However, it does define a representation of the maximal torus $ R:\prod_x U(1)^{K(x)}\to \prod_{\langle x;y\rangle} U(K(x) K(y) )$. The monopole formula involves the quotient of $T = \prod_x U(1)^{K(x)}$ by $\operatorname{Ker}(R)_0$, the connected part of the kernel of $R$. 

The kernel $\operatorname{Ker}(R)$ is included in $ D= \prod_x U(1)_{x,\rm{diag}}$ where $U(1)_{x,\rm{diag}}$ is the diagonal subgroup of $U(K(x))$, and the restriction $R|_D$ is the morphism associated with the Abelian quiver obtained from $Q$ by setting all weights to $1$. As $Q$ is stable and non-Abelian, this quiver must be free. This implies, in particular, that the kernel $\operatorname{Ker}(R)$ is connected. The monopole formula for $Q$ therefore only depends on the image of $R$.

Likewise, because of the assumption on $\ell$, we see that $\operatorname{Ker}(R')$ is connected and the monopole formula for $Q'$ therefore only depends on the image of $R'$. To conclude, notice that we have a commutative diagram: 
\begin{center}
    \begin{tikzcd}
    D \ar[r] \ar[dr,"R'"']& D \ar[d,"R"]\\
    & \prod_{\langle x;y\rangle} U(K(x) K(y))
    \end{tikzcd}
\end{center}
The horizontal arrow is the morphism obtained by applying $z \in U(1)\mapsto z^\ell \in U(1)$ to the node $x$. In particular, this is a surjective morphism. Therefore, $\operatorname{Im}(R) = \operatorname{Im}(R')$ and both quivers have the same Coulomb branch. In particular, $Q'$ is a stable quiver.
\end{proof}

Note that in the physics literature, this construction can be understood via framing/unframing operations of the gauge group encoded in the quiver. This is also known as Crawley-Boevey moves \cite{Crawley-Boevey1} in the mathematics literature.

\subsection{Non-simply-laced Quivers}
\label{sec:general_case}
We now extend the classification to include quivers with non-simply-laced edges. Our main result, Theorem~\ref{thm:final_result}, is a direct consequence of the two following propositions. 

\begin{proposition} \label{prop:list_stable_non_abelian_non_simply_laced_quivers}
The quivers in Tables~\ref{tab:results} and \ref{tab:dynkin} are quasi-good. 
\end{proposition}
\begin{proof}
    The Hilbert series for each quiver in Tables~\ref{tab:results} and \ref{tab:dynkin} is known, cf.\ the references provided.
\end{proof}

\vspace{4pt}
\begin{proposition}
\label{thm:non-simply_laced}
Let $Q$ be a quasi-good quiver. Then $Q$ admits a decay product isomorphic to one of the quivers listed in Tables~\ref{tab:results} and \ref{tab:dynkin}. 
\end{proposition}

\begin{proof}
Let $Q$ be a quasi-good quiver.  Any connected component of $Q$ is a decay product, so we only have to deal with the case where $Q$ is connected.

If $Q$ is Abelian or admits a quasi-good Abelian decay product (or even a non-free Abelian subquiver), then we can apply Theorem~\ref{thm:list_stable_abelian_quivers}. Therefore, we can assume that $Q$ has only free Abelian subquivers. In particular, $Q$ is non-Abelian, has no cycles\footnote{The existence of a length function is required here, as it plays a crucial role in the proof \cite{Bourget:2024asp} of Theorem~\ref{thm:list_stable_abelian_quivers}.} and if two non-simply-laced edges are directed towards each other, then their lacednesses are coprime. 

Just as in the simply-laced case, we can construct a decay product of $Q$ where all nodes of weight $1$ have exactly one neighbor. Since the result for this decay product implies the result for $Q$, we can assume that $Q$ has this property.

Let us call \emph{trivial} a non-simply-laced edge whose source is a node of rank $1$. As discussed in Section~\ref{subsec:trivial_edges}, the field theory associated with $Q$, and in particular its Coulomb branch, is the same as the one of $Q'$, obtained by replacing all trivial non-simply-laced edges by simply-laced ones. Explicitly, if we prove the theorem for $Q'$ and find a decay $Q' \rightsquigarrow Q''$, then $Q''$, after adding back trivial edges if needed, is a decay product of $Q$. We can therefore restrict to the case where $Q$ has no trivial non-simply-laced edges.

If $Q$ has an edge of lacedness greater than or equal to $4$, then it decays to $\mathcal{Y}(\ell)$ (cf.\ Table~\ref{tab:results}). Therefore we can assume that all non-simply-laced edges have lacedness $2$ or $3$.

If $Q$ has an edge of lacedness $3$ whose source has two or more neighbors, then $Q$ decays to $G_2^{(1)}$ (cf. Table~\ref{tab:dynkin}), $\mathcal J_{2,3}$ or $\mathcal J_{3,3}$ (cf. Table~\ref{tab:results}). Therefore, we can assume that any triple edge is sourced at a node of $Q$ with exactly one neighbor. 

If $Q$ has two or more non-simply-laced edges, then, by considering two such edges such that all edges in between are simply-laced, we see that $Q$ decays to a $\tilde{B}_n^{(2)}$ (or $A_{2n}^{(2)}$) (cf.\ Table~\ref{tab:dynkin}), $B_n^{(2)}$ (or $D_{n+1}^{(2)}$) (cf.\ Table~\ref{tab:dynkin}), $gc_n$ (cf.\ Table~\ref{tab:results}), or $gb_n$ (cf.\ Table~\ref{tab:results}).

If it has only simply-laced edges, Theorem~\ref{thm:simply_laced_decay} applies.
Therefore, we can assume that $Q$ has exactly one non-simply-laced edge.

If $Q$ has a vertex of degree greater than or equal to $3$, then $Q$ decays to  $C_n^{(2)}$ (or $A_{2n-1}^{(2)}$) (Table~\ref{tab:dynkin}), $B_n^{(1)}$ (Table~\ref{tab:dynkin}), or $gd_n$ (Table~\ref{tab:results}). Therefore, $Q$ is a linear quiver with exactly one non-simply-laced edge of lacedness $\ell \in \{2,3\}$.
If $\ell =3$, then $Q$ is shaped as in \begin{align}
    \raisebox{-.5\height}{
\begin{tikzpicture}
    \tikzset{node distance = 1cm};
    \node (gg1) [gauge,label=below:{$a$}] {};
    \node (gg2) [gauge,right of=gg1,label=below:{$b_0$}] {};
    \node (gg3) [gauge,right of=gg2,label=below:{$b_1$}] {};
    \node (gg4) [right of=gg3] {$\cdots$};
    \draw (gg1)--(gg2) (gg2)--(gg3) (gg3)--(gg4) ;
    \draw[double distance=1.5pt,scaling nfold=3,decoration={markings, mark=at position 0.65 with {\arrow[scale=0.7]{>}}}, postaction={decorate}] (gg1) to (gg2);
\end{tikzpicture}
}
\label{eq:proof_Q2_d4}
\end{align}
Let $a$ be the weight of the long node and $b_0,b_1,\ldots$ the weight of the short nodes. The balance conditions read:
\begin{subequations}
\begin{align}
2a &\leq 3 b_0 \, ,\label{eq:proof_d4_balance1}\\
2b_0 &\leq a +b_1 \,, \label{eq:proof_d4_balance2}\\
2b_n &\leq b_{n-1} + b_{n+1} \quad (\forall n \geq 1) \,.
\end{align}
\end{subequations}
The sequence $(b_n)_{n\geq 0}$ is convex and vanish for $n$ large enough. Therefore, it is positive and strictly decreasing before it reaches $0$. In particular, $b_1 \leq b_0 - 1$. Together with equations (\ref{eq:proof_d4_balance1}) and (\ref{eq:proof_d4_balance2}), this implies $a\geq 3$, $b_0\geq 2$ and $b_1 \geq 1$. Hence $Q$ decays to $D_4^{(3)}$ ($G_2^{(3)}$)  (cf. Table~\ref{tab:dynkin}).

We are only left with the case $\ell = 2$. If a long node has three neighbors, then $Q$ decays to a $B_n^{(1)}$ quiver. Therefore, we can assume that $Q$ is a linear quiver
\begin{align}
    \raisebox{-.5\height}{
\begin{tikzpicture}
    \tikzset{node distance = 1cm};
    \node (gg1) [gauge,label=below:{$a_1$}] {};
    \node (gg2) [gauge,right of=gg1,label=below:{$a_0$}] {};
    \node (gg3) [gauge,right of=gg2,label=below:{$b_0$}] {};
    \node (gg4) [gauge,right of=gg3,label=below:{$b_1$}] {};
    \node (lhs) [left of=gg1] {$\cdots$};
    \node (rhs) [right of =gg4] {$\cdots$};
    \draw (gg1)--(gg2) (gg3)--(gg4) (gg1)--(lhs) (gg4)--(rhs);
     \draw[double distance=1.5pt,scaling nfold=2,decoration={markings, mark=at position 0.65 with {\arrow[scale=0.7]{>}}}, postaction={decorate}] (gg2) to (gg3);
\end{tikzpicture}
}
\label{eq:proof_Q2_F}
\end{align}
Let us write $a_0,a_1,\ldots $ for the weights of the long nodes and $b_0,b_1,\ldots$ for the weights of the short nodes. We already know that $a_0 \geq 2,b_0\geq 1$. The balance conditions read, with $n\in \mathbb{Z}_{>0}$: 
\begin{subequations}
\begin{align}
2a_0 &\leq 2 b_0 + a_1 \,, \label{eq:2long_1}\\
2b_0 &\leq a_0+b_1 \,,\label{eq:2long_2} \\
2a_n &\leq a_{n-1} + a_{n+1}\,, \label{eq:2long_3}\\
2b_n &\leq b_{n-1} + b_{n+1} \,.\label{eq:2long_4}
\end{align}
\end{subequations}
As before, the sequences $(a_n)_{n\in\mathbb N}$ and $(b_n)_{n\in\mathbb N}$ are convex, positive and strictly decreasing before they reach $0$. In particular $b_n \leq b_0 - n$ and $a_n\leq a_0 - n$. This inequalities for $n=1$, together with \eqref{eq:2long_1} and  \eqref{eq:2long_2} imply that:
\begin{subequations}
\begin{align}
a_0 +1 \leq 2b_0 \,,\\
1+b_0 \leq a_0 \,.\label{eq:2long_6}
\end{align}
\end{subequations}
All the solutions of this system of inequalities satisfy $a_0\geq 3$ and $b_0\geq 2$.
Then, from (\ref{eq:2long_1}) and  (\ref{eq:2long_2}), we also see that $a_1\geq 2$ and $b_1\geq 1$. 

If $a_2 \geq 1$, we see that $Q$ decays to $F_4^{(1)}$ as $a_1\geq2$ (cf. Table~\ref{tab:dynkin}), so we can now assume that $a_2=0$. Equation  \eqref{eq:2long_3} reduces to $2a_1 \leq a_0$.

If $b_2 = 0$, we have  $2b_1  \leq b_0$. Together with \eqref{eq:2long_1} and \eqref{eq:2long_2}, we get $\frac34 a_0 \leq b_0 \leq \frac23 a_0$ which is a contradiction as $a_0 >0$. 
Therefore $b_2\geq 1$. This implies $b_1\geq 2$, $b_0\geq 3$, and, using \eqref{eq:2long_6},  $a_0\geq 4$. 
We conclude that $Q$ decays to $F_4^{(2)}$ (or $E_6^{(2)}$) (cf. Table~\ref{tab:dynkin}).
\end{proof}

\subsection{HWG Computations}
\label{sec:HWG_results}
Part of the insight gained in this work are the three new families of unitary quivers (cf.\ Table~\ref{tab:results}) --- $gb_n$, $gc_n$, $gd_n$. As a first analysis of the geometry, the (unrefined) Hilbert Series (see Definition~\ref{def:HS}) for the lowest dimensional cases are calculated to be  
\begin{align}
    \HS_{gb_{2}}(t)=&\frac{1+6t^{2}+42t^{4}+71t^{6}+122t^{8}+71t^{10}+42t^{12}+6t^{14}+t^{16}}{\left(1-t^{2}\right)^{4}\left(1-t^{4}\right)^{4}} 
    \\
    \HS_{gc_{2}}(t)=&\frac{1+20t^{2}+175t^{4}+590t^{6}+1290t^{8}+1550t^{10}+1290t^{12}+590t^{14}+175t^{16}+20t^{18}+t^{20}}{\left(1-t^{2}\right)^{5}\left(1-t^{4}\right)^{5}}
    \\
    \HS_{gd_{3}}(t)=&\frac{1+10t^{2}+85t^{4}+239t^{6}+545t^{8}+602t^{10}+545t^{12}+239t^{14}+85t^{16}+10t^{18}+t^{20}}{\left(1-t^{2}\right)^{5}\left(1-t^{4}\right)^{5}}
\end{align}
In addition, for the lowest $n$ cases we computed the refined Hilbert series, which allows to derive the Highest Weight Generating function (recall Definition~\ref{def:HWG}). The obtained HWGs are then expected to generalize into the full $gb_n$, $gc_n$, $gd_n$ families as shown in Table~\ref{tab:HWG}. For higher $n$ cases, this proposed HWG has been tested against unrefined Hilbert series computations. Note the remarkable fact that the HWG for $gb_n$ and $gd_n$ are exactly the same (even though the symmetry algebra is not the same), and they \emph{do not depend on $n$}. 

The HWG for $gc_n$ coincides with that of $\overline{h}_{2n+1,\sigma = (3,1,\dots,1)}$ \cite[(3.18)]{Bourget:2021siw}.\footnote{We thank Paul Levy for pointing this out to us. } This leads to the conjecture that the two geometries are the same. If this is true, the $gc_n$ still should appear in our result as a stable quiver, but the associated Coulomb branch is a particular case of the $\overline{h}_{n,\sigma}$ family:
\begin{equation}\label{eq:CBequality}
   \mathcal{M}_C \left( \raisebox{-.5\height}{\begin{tikzpicture} 
            \tikzset{node distance=1cm};
            \node (1) at (-1,0) [gauge,label={[align=center]below:\small{$1$}}] {};
            \node (2) at (-2,0) [gauge,label={[align=center]below:\small{$1$}}] {};
            \node (22) at (-3,0) {$\cdots$};
            \node (3) at (-4,0) [gauge,label={[align=center]below:\small{$1$}}] {};
            \node (4) at (-5,0) [gauge,label={[align=center]below:\small{$1$}}] {};
            \node (1a) at (-1,1) [gauge,label={[align=center]above:\small{$1$}}] {};
            \node (2a) at (-2,1) [gauge,label={[align=center]above:\small{$1$}}] {};
            \node (22a) at (-3,1) {$\cdots$};
            \node (3a) at (-4,1) [gauge,label={[align=center]above:\small{$1$}}] {};
            \node (4a) at (-5,1) [gauge,label={[align=center]above:\small{$1$}}] {};
            \draw[double distance=1.5pt,scaling nfold=3, decoration={markings, mark=at position 0.65 with {\arrow[scale=0.7]{<}}}, postaction={decorate}] (3)to (4);
            \draw[double distance=1.5pt,scaling nfold=3, decoration={markings, mark=at position 0.65 with {\arrow[scale=0.7]{<}}}, postaction={decorate}] (3a)to (4a);
            \draw (3)--(22)--(2)--(1)--(1a)--(2a)--(22a)--(3a) (4)--(4a);
            \end{tikzpicture}} \right) =    \mathcal{M}_C \left( \raisebox{-.5\height}{
            \begin{tikzpicture} 
\tikzset{node distance=1cm};
\node (0) at (0,0) [gauge,label={[align=center]above:\small{$2$}}] {};
\node (1) at (1,0) [gauge,label={[align=center]above:\small{$2$}}] {};
\node (2) at (2,0)  {$\cdots$};
\node (3) at (3,0) [gauge,label={[align=center]above:\small{$2$}}] {};
\node (4) at (4,0) [gauge,label={[align=center]above:\small{$2$}}] {};
\draw[double distance=1.5pt,scaling nfold=3,decoration={markings, mark=at position 0.65 with {\arrow[scale=0.7]{>}}}, postaction={decorate}] (0) to (1);
\draw (1)--(2)--(3);
\draw[double distance=1.5pt,scaling nfold=2,decoration={markings, mark=at position 0.65 with {\arrow[scale=0.7]{>}}}, postaction={decorate}] (4) to (3);
\end{tikzpicture}} \right) \, . 
\end{equation}

\section{\texorpdfstring{Generalization to $(p,q)$-edges}{Generalization to pq-edges}}
\label{sec:pq_edge}
Let us relax assumption (iii) in Definition~\ref{def:quiver} to allow for quivers with $(p,q)$-edges.
\begin{definition} \label{def:quiver_p-q_edge}
A \textbf{generalized quiver} $Q$ is a triple $(V,A,K)$ where $V$ is a finite set, $A$ a function $V\times V\to \mathbb{Z}$ and $K$ a function $V\to \mathbb Z_{>0}$, such that
\begin{assertions}
	\item for all $x \in V$, $A(x,x) = -2+2g_x$, with $g_x \in \mathbb{Z}_{\geq 0}$. If $K(x) = 1$, then $A(x,x)=-2$.
    \item there exists a function $L: V\to\mathbb Z_{>0}$ such that for every $x,y\in V$, $A(x,y) L(y)= A(y,x)L(x) $.
\end{assertions}
\end{definition}

\begin{rmq}
    A $(p,q)$-edge is the extension to the case that if $A(x,y)>0$ for some $x,y\in V$, $A(x,y)$ does not need to be a divisor of $A(y,x)$, and vice versa. Such an edge is called a $(p,q)$-edge with $p=A(x,y)$ and $q=A(y,x)$ \cite{Grimminger:2024doq}.  The Cartan matrix $A$ is symmetrizable, due to the length function.
\end{rmq}
Based on analysis of the abelian generalized quivers in \cite{Bourget:2024asp}, it is natural to propose the following:
\begin{conjecture}[Decay and Fission for generalized quivers]\label{conj:Decay_Fission_p-q}
    The Decay and Fission algorithm holds true for quasi-good quivers of Definition~\ref{def:quiver_p-q_edge}.
\end{conjecture}

\begin{ftheo}
    Let $Q$ be a generalized quiver, then it admits a decay product isomorphic to a quiver listed in Tables~\ref{tab:results} and \ref{tab:dynkin}, if one allows for $gc_1$.
\end{ftheo}
\begin{proof}
   The proof requires an extension of the proof of Theorem~\ref{thm:non-simply_laced}, assuming Conjecture~\ref{conj:Decay_Fission_p-q} holds. 

First, the $A_{N-1}$ quiver is generalized to $(p,q)$-edges with multiplicity $N$
\begin{align}
    \begin{tikzpicture}[baseline=0,xscale=1.5]
        \node (0) at (0,0) [gauge,label={[align=center]below:\small{$1$}}] {};
        \node (1) at (1,0) [gauge,label={[align=center]below:\small{$1$}}] {};
        \draw[line width=2pt,darkgray, decoration={markings, mark=at position 0.2 with {\arrow[scale=0.7]{>}}, mark=at position 0.9 with {\arrow[scale=0.7]{<}}}, postaction={decorate},transform canvas={xshift=0pt,yshift=4.75pt}] (0) to node[midway, above, black] {\small{$(p,q)$}} (1);
        \draw[line width=2pt,darkgray, decoration={markings, mark=at position 0.2 with {\arrow[scale=0.7]{>}}, mark=at position 0.9 with {\arrow[scale=0.7]{<}}}, postaction={decorate},transform canvas={xshift=0pt,yshift=-4.75pt}] (0) to node[midway, above, black] {} (1);
        \node at (0.5,0.095) [label={[align=center]right:\small{}}] {\footnotesize{${\scriptstyle{\vdots\,}} N$}};
        \end{tikzpicture} \label{eq:pq_with_multiplicity}
\end{align}
whose Coulomb branch is $A_{N\cdot \gcd(p,q)-1}$. Thus, \eqref{eq:pq_with_multiplicity} is the only stable quiver with edge-multiplicity larger than $1$. 
Second, we realize that the $\mathcal{Y}(\ell)$ quiver is naturally extended to the $(p=\ell^\prime,q=\ell)$-edges cases,
\begin{align}
\begin{tikzpicture}[baseline=0]
         \tikzset{node distance=1cm};
            \node (1a) [gauge,label={[align=center]below:\small{$1$}}] {};
            \node (1b) [gauge,right of=1a,label={[align=center]below:\small{$2$}}] {};
            \draw[line width=2pt,darkgray, decoration={markings, mark=at position 0.65 with {\arrow[scale=0.7]{>}}}, postaction={decorate}] (1b)to node[midway,above,black] {\small{$\ell$}} (1a);
            \end{tikzpicture}
            \qquad \leftrightarrow\qquad
        \begin{tikzpicture}[baseline=0,xscale=1.5]
        \node (0) at (0,0) [gauge,label={[align=center]below:\small{$1$}}] {};
        \node (1) at (1,0) [gauge,label={[align=center]below:\small{$2$}}] {};
        \draw[line width=2pt,darkgray, decoration={markings, mark=at position 0.2 with {\arrow[scale=0.7]{>}}, mark=at position 0.9 with {\arrow[scale=0.7]{<}}}, postaction={decorate}] (0) to node[midway, above, black] {\small{$(\ell^\prime,\ell)$}} (1);
        \end{tikzpicture}
\end{align}
analogous to Section~\ref{subsec:trivial_edges}. Therefore, we can restrict to $(p,q)$-edges between non-Abelian vertices with $p,q<4$. 

Next, by the linear Abelian quivers $h_{n=1,\delta,\sigma}$ (see Table~\ref{tab:results}), we can restrict to co-prime $p,q$, as otherwise a decay to an Abelian quiver exists. 
Thus, without loss of generality, we are led to a $(p,q)$-edge with $p=3$ and $q=2$. This is precisely $gc_1$, for which one computes
\begin{align}
    gc_1:  \begin{tikzpicture}[baseline=0,xscale=1.5]
        \node (0) at (0,0) [gauge,label={[align=center]below:\small{$2$}}] {};
        \node (1) at (1,0) [gauge,label={[align=center]below:\small{$2$}}] {};
        \draw[line width=2pt,darkgray, decoration={markings, mark=at position 0.2 with {\arrow[scale=0.7]{>}}, mark=at position 0.9 with {\arrow[scale=0.7]{<}}}, postaction={decorate}] (0) to node[midway, above, black] {\small{$(3,2)$}} (1);
        \end{tikzpicture}
        \;,\qquad 
        \HS_{gc_1}(t)=
    \frac{1+6t^2+29t^4+30t^6+29t^8+6t^{10}+t^{12}}{(1-t^2)^3(1-t^4)^3} \,,
\end{align}
which is compatible with the $n=1$ limit of $\HWG_{gc_n}$ of Table~\ref{tab:HWG}.
This concludes the proof.
\end{proof}

\appendix

\section{Hilbert Series and Highest Weight Generating Functions}
\label{app:HS_and_HWG}

In general, the $\mathbb C^*$-action on a CSS induces a grading on the coordinate ring, and the dimensions of the subspaces of fixed degree are finite. The \textbf{Hilbert series} is defined as their generating function.  

\begin{definition} \label{def:HS}
    For a quiver $Q=\left(V,A,K\right)$ the \textbf{conformal dimension} is defined as the function  
    \begin{equation}
        \begin{array}{crcl}
           \Delta_{Q}: & \Lambda\equiv\prod_{x\in V}\mathbb{Z}^{K(x)}  & \rightarrow & \mathbb{R}  \\ & (\vec{m}_x)_{x \in V} & \mapsto & \frac{1}{2} \sum\limits_{x,y \in V}  \mathrm{sgn} (A(x,y)) \delta ( \vec{m}_x A(x,y)  ,  \vec{m}_y A(y,x) ) \\
             & 
        \end{array}
    \end{equation}
where $\delta (\vec{u} , \vec{v}) = \sum\limits \notag|u_i - v_j|$. 
From the length function $L$ we define $\textbf{L} \in \prod_{x\in V}\mathbb{Z}^{K(x)}$ as $\textbf{L}|_{\mathbb{Z}^{K(x)}} = L(x) \cdot (1,\ldots,1)$ where $(1,\ldots,1) \in \mathbb{Z}^{K(x)}$, such that $\Delta_{Q}(\textbf{m}+\textbf{L})=\Delta_{Q}(\textbf{m})$.

Let $W = \prod_{x \in V} \mathfrak{S}_{K (x)}$ where $\mathfrak{S}_n$ is the symmetric group on $n$ objects. 
The \textbf{Hilbert series} of the Coulomb branch is given \cite{Nakajima:2019olw}, if the series converges, by the \emph{Monopole Formula} \cite{Cremonesi:2013lqa,Bourget:2020bxh}
\begin{equation}
\HS_{Q} (t) = \frac{1-t^2}{|W|} \sum\limits_{\mathbf{m} \in \Lambda / \mathbf{L} \mathbb{Z}} \; \sum\limits_{w \in W (\mathbf{m})} \frac{t^{\Delta_Q (\mathbf{m})}}{\det (1 - w t^2)} \, . 
\end{equation}
where $w$ is seen as a permutation matrix and $W(\mathbf{m}) = \{ w \in W | w \cdot \mathbf{m} = \mathbf{m}\}$. 
\end{definition}


\begin{rmq} When the Coulomb branch has a symmetry algebra $\mathfrak{g}$, one can \emph{refine} the Hilbert series, which becomes a power series in $t$ with coefficients which are characters of $\mathfrak{g}$, written as Laurent polynomials in $\mathrm{rank}(\mathfrak{g})$ variables which we denote as $z_1 , \dots , z_{\mathrm{rank}(\mathfrak{g})}$. Then the refined Hilbert series can be written in a unique way as
\begin{equation}
\label{eq:refinedHS}
\HS_{Q}(t ; z_1 , \dots , z_r ) = \sum\limits_{n \in \mathbb{N}} \sum\limits_{n_1 , \dots , n_r \in \mathbb{N}} a_{n,n_1 , \dots n_r} \chi_{[n_1 , \dots , n_r]} (z_1 , \dots , z_r ) t^n \, ,  
\end{equation}
where $\chi_{[n_1 , \dots , n_r]} (z_1 , \dots , z_r )$ is the character for the representation of $\mathfrak{g}$ with highest weight specified by Dynkin labels $[n_1 , \dots , n_r]$. For the ordering of the labels, we choose conventions such that for $\mathfrak{g}= \mathfrak{su}_n$, $n \geq 3$, $[1,0,\dots,0]$ corresponds to the fundamental representation and $[1,0,\dots,0,1]$ to the adjoint representation.  For $\mathfrak{g}=\mathfrak{so}_n$, $[1,0,\dots,0]$ corresponds to the fundamental representation and $[0,1,\dots,0]$ to the adjoint representation for $n \geq 7$. For $\mathfrak{so}_6$, the adjoint is $[0,1,1]$, for $\mathfrak{so}_5$, the adjoint is $[0,2]$.   
\end{rmq}

\begin{definition}\label{def:HWG}
    The \textbf{Highest Weight Generating} (HWG) function \cite{Hanany:2014dia} associated to \eqref{eq:refinedHS} is defined to be $\HWG_{Q} \in \mathbb{C}[[t , \mu_1 , \dots , \mu_r]]$ with 
\begin{equation}
\HWG_{Q} = \sum\limits_{n \in \mathbb{N}} \sum\limits_{n_1 , \dots , n_r \in \mathbb{N}} a_{n,n_1 , \dots n_r}  t^n \mu_1^{n_1} \cdots \mu_r^{n_r} \, . 
\end{equation}
\end{definition} 

The plethystic logarithm, and its inverse the plethystic exponential, are natural operations when studying generating functions.

\begin{definition}
\label{def:PL}
    For $f \in \mathbb{C}[[t_1 , \dots , t_k]]$ such that $f(0,\dots,0)=1$, 
    the \textbf{Plethystic Logarithm} (PL) is defined to be
    \begin{equation}
    \PL[f](t_1 , \dots , t_k) = \sum\limits_{j=1}^{\infty} \frac{\mu (j)}{j} \log (f(t_1^j \dots , t_k^j)) \, . 
    \end{equation}
    Here, $\mu$ is the Möbius multiplicative function. 
\end{definition}

For some ICSSs, such as the minimal nilpotent orbits or the Coulomb branches of the $gb_n,gc_n$ and $gd_n$ quivers, the plethystic logarithm of the HWG turns out to be a simple polynomial.

\bibliographystyle{JHEP}     
\bibliography{references}

\end{document}

%% file: table_results.tex
\begin{table}
\centering
\hspace*{-1.0cm}
\begin{NiceTabular}{cccc}
\midrule[2pt]
& Geometry & Quiver & Condition \\ \midrule[2pt]
\Block{3-1}<\rotate>{\large{\textsc{Abelian Quivers}}} 
& $A_{N-1}$ & \begin{tikzpicture}[baseline=0][baseline=0]
                \tikzset{node distance=2cm};
                \node (1a) at (0,0) [gauge,label={[align=center]below:\small{$1$}}] {};
                \node (1b) [gauge,right of=1a,label={[align=center]below:\small{$1$}}] {};
                \draw[line width=2pt,darkgray, decoration={markings, mark=at position 0.6 with {\arrow[scale=0.7]{>}}},postaction={decorate}] (1a)to node[midway, above, black] {\small{$\ell$, $N$}} (1b);
                \end{tikzpicture} & \small{{$\!\begin{aligned}&N\geq2, \, \ell\geq1\\&\textrm{$N$ copies of $\ell$-edge}\end{aligned}$}} \\ \cmidrule{2-4}
& $h_{n , \delta , \sigma}$ & \raisebox{-.5\height}{\begin{tikzpicture}[xscale=1.5]
        \node (0) at (0,0) [gauge,label={[align=center]below:\small{$1$}}] {};
        \node (1) at (1,0) [gauge,label={[align=center]below:\small{$1$}}] {};
        \node (2) at (1.3,0) {$\cdots$};
        \node (3) at (1.6,0) [gauge,label={[align=center]below:\small{$1$}}] {};
        \node (4) at (2.6,0) [gauge,label={[align=center]below:\small{$1$}}] {};
        \draw[line width=2pt,darkgray, decoration={markings, mark=at position 0.2 with {\arrow[scale=0.7]{>}}, mark=at position 0.9 with {\arrow[scale=0.7]{<}}}, postaction={decorate}] (0) to node[midway, above, black] {\small{$(\ell_{1},k_{1})$}} (1);
        \draw[line width=2pt,darkgray, decoration={markings, mark=at position 0.2 with {\arrow[scale=0.7]{>}}, mark=at position 0.9 with {\arrow[scale=0.7]{<}}}, postaction={decorate}] (3) to node[midway, above, black] {\small{$(\ell_{n},k_{n})$}} (4);
        \end{tikzpicture}} & \small{{$\!\begin{aligned}& n \geq 2 \\ &\delta\equiv\textrm{gcd}\left(\ell_{1},k_{n}\right)>1\\&\textrm{gcd}\left(\ell_{i},k_{j}\right)=1 \; \textrm{for all other} \; 1\leq i\leq j\leq n\\&\textrm{Charge vector} \; \sigma=\left(\sigma_{1},...,\sigma_{n}\right)\in\left(\mathbb{Z}^{*}_{\delta}\right)^{n}\\&\textrm{with}\; \sigma_{i}/\sigma_{i-1}=-\ell_{i}/k_{i-1}\;\textrm{mod}\;\delta\;\forall\;2\leq i\leq n \end{aligned}$}} \\  \cmidrule{2-4}  
&  $\overline{h}_{n , \sigma}$      &  \raisebox{-.5\height}{\begin{tikzpicture} 
    \tikzset{node distance=2cm};
    \node (1a) at (0,0) [gauge,label={[align=center]above:\small{$1$}}] {};
    \node (1b) [gauge,right of=1a,label={[align=center]above:\small{$1$}}] {};
    \node (1c) at (-1,-1.5) [gauge,label={[align=center]left:\small{$1$}}] {};
    \node (1d) at (3,-1.5) [gauge,label={[align=center]right:\small{$1$}}] {};
    \node (1e) at (0,-3) [gauge,label={[align=center]below:\small{$1$}}] {};
    \node (1f) at (2,-3) [gauge,label={[align=center]below:\small{$1$}}] {};
    \node (End1) at (0.5,-3) {};
    \node (End2) at (1.5,-3) {};
    \draw[line width=2pt,darkgray, decoration={markings, mark=at position 0.2 with {\arrow[scale=0.7]{>}}, mark=at position 0.9 with {\arrow[scale=0.7]{<}}}, postaction={decorate}] (1a)to node[midway,above,black,rotate=0] {\small{$(k_{1},\ell_{1})$}} (1b);
    \draw[line width=2pt,darkgray, decoration={markings, mark=at position 0.2 with {\arrow[scale=0.7]{>}}, mark=at position 0.9 with {\arrow[scale=0.7]{<}}}, postaction={decorate}] (1b)to node[midway,above,black,rotate=-55] {\small{$(k_{2},\ell_{2})$}} (1d);
    \draw[line width=2pt,darkgray, decoration={markings, mark=at position 0.2 with {\arrow[scale=0.7]{>}}, mark=at position 0.9 with {\arrow[scale=0.7]{<}}}, postaction={decorate}] (1a)to node[midway,above,black,rotate=58] {\small{$(k_{n+1},\ell_{n+1})$}} (1c);
    \draw[line width=2pt,darkgray, decoration={markings, mark=at position 0.2 with {\arrow[scale=0.7]{>}}, mark=at position 0.9 with {\arrow[scale=0.7]{<}}}, postaction={decorate}] (1c)to node[midway,below,black,rotate=-55] {\small{$(\ell_{n},k_{n})$}} (1e);
    \draw[line width=2pt,darkgray, decoration={markings, mark=at position 0.2 with {\arrow[scale=0.7]{>}}, mark=at position 0.9 with {\arrow[scale=0.7]{<}}}, postaction={decorate}] (1d)to node[midway,below,black,rotate=55] {\small{$(\ell_{3},k_{3})$}} (1f);
    \draw[line width=2pt,darkgray, decoration={markings, mark=at position 0.9 with {\arrow[scale=0.7]{>}}}, postaction={decorate}] (1e)to (End1);
    \draw[line width=2pt,darkgray, decoration={markings, mark=at position 0.9 with {\arrow[scale=0.7]{>}}}, postaction={decorate}] (1f)to (End2);
    \draw[line width=1pt,darkgray,dotted] (End1)to (End2);
    \end{tikzpicture}} & \small{$\!{\begin{aligned}& n\geq3 \\&\prod_i k_i = \prod_j \ell_j \; \textrm{(length-function)}\\&\textrm{gcd}(\ell_i , k_j)=1 \; \forall\; (i,j)\in\mathbb{Z}_{n+1}^{2}\; \\& \textrm{such that}\;\; i-j \not\equiv 1,2 \;\textrm{mod}\; n+1\\&\textrm{Charge vector}\;\sigma=\left(\sigma_{1},...,\sigma_{n+1}\right)\in\mathbb{Z}^{n+1}\\&\textrm{with}\;\sigma_{i}=\textrm{gcd}\left(\ell_{i+2},k_{i}\right)\;\forall\; i \in \mathbb{Z}_{n+1} \end{aligned}}$} \\ \midrule[2pt]
& $\overline{\mathcal{O}_{\text{min}}} (\mathfrak{g})$ & \begin{tabular}{c}
     Affine Twisted / Untwisted Balanced  \\
    Dynkin Quiver for $\mathfrak{g}$
\end{tabular}  & see Table~\ref{tab:dynkin} \\ \midrule 
\Block{7-1}<\rotate>{\large{\textsc{Non-Abelian Quivers}}} & $\begin{cases}
    A_1 & (g=1) \\ D_{g+1} & (g \geq 2)
\end{cases}$ & \begin{tikzpicture}[baseline=0]
            \tikzset{node distance=0.5cm};
            \node (1a) [gauge,label={[align=center]below:\small{$2$}}] {};
            \node (1b) [above of=1a,scale=0.001,label={[align=center]above:\small{$g$}}] {};
            \draw[line width=1pt,darkgray] (1a) to[out=45,in=0] (1b);
            \draw[line width=1pt,darkgray] (1b) to[out=180,in=135] (1a);
            \end{tikzpicture} & \small{$g \geq 1$} \\ \cmidrule{2-4}
& $\mathcal{Y}(\ell)$ & \begin{tikzpicture}[baseline=0]
            \tikzset{node distance=1cm};
            \node (1a) [gauge,label={[align=center]above:\small{$1$}}] {};
            \node (1b) [gauge,right of=1a,label={[align=center]above:\small{$2$}}] {};
            \draw[line width=2pt,darkgray, decoration={markings, mark=at position 0.65 with {\arrow[scale=0.7]{>}}}, postaction={decorate}] (1b)to node[midway,above,black] {\small{$\ell$}} (1a);
            \end{tikzpicture} & \small{$\ell\geq4$} \\ \cmidrule{2-4}
&  $gc_n$   & \begin{tikzpicture}[baseline=0][baseline=0]
\tikzset{node distance=1cm};
\node (0) at (0,0) [gauge,label={[align=center]above:\small{$2$}}] {};
\node (1) at (1,0) [gauge,label={[align=center]above:\small{$2$}}] {};
\node (2) at (2,0)  {$\cdots$};
\node (3) at (3,0) [gauge,label={[align=center]above:\small{$2$}}] {};
\node (4) at (4,0) [gauge,label={[align=center]above:\small{$2$}}] {};
\draw[double distance=1.5pt,scaling nfold=3,decoration={markings, mark=at position 0.65 with {\arrow[scale=0.7]{>}}}, postaction={decorate}] (0) to (1);
\draw (1)--(2)--(3);
\draw[double distance=1.5pt,scaling nfold=2,decoration={markings, mark=at position 0.65 with {\arrow[scale=0.7]{>}}}, postaction={decorate}] (4) to (3);
\end{tikzpicture} & \small{$n\geq2$} \\ \cmidrule{2-4}
 & $gb_n$   & \begin{tikzpicture}[baseline=0][baseline=0]
\tikzset{node distance=1cm};
\node (0) at (0,0) [gauge,label={[align=center]above:\small{$2$}}] {};
\node (1) at (1,0) [gauge,label={[align=center]above:\small{$2$}}] {};
\node (2) at (2,0)  {$\cdots$};
\node (3) at (3,0) [gauge,label={[align=center]above:\small{$2$}}] {};
\node (4) at (4,0) [gauge,label={[align=center]above:\small{$1$}}] {};
\draw[double distance=1.5pt,scaling nfold=3,decoration={markings, mark=at position 0.65 with {\arrow[scale=0.7]{>}}}, postaction={decorate}] (0) to (1);
\draw (1)--(2)--(3);
\draw[double distance=1.5pt,scaling nfold=2,decoration={markings, mark=at position 0.65 with {\arrow[scale=0.7]{>}}}, postaction={decorate}] (3) to (4);
\end{tikzpicture}  & \small{$n\geq2$} \\ \cmidrule{2-4}
 &   $gd_n$   & \begin{tikzpicture}[baseline=0][baseline=0]
\tikzset{node distance=1cm};
\node (0) at (0,0) [gauge,label={[align=center]above:\small{$2$}}] {};
\node (1) at (1,0) [gauge,label={[align=center]above:\small{$2$}}] {};
\node (2) at (2,0)  {$\cdots$};
\node (3) at (3,0) [gauge,label={[align=center]above:\small{$2$}}] {};
\node (4a) at (3.7,.4) [gauge,label={[align=center]right:\small{$1$}}] {};
\node (4b) at (3.7,-.4) [gauge,label={[align=center]right:\small{$1$}}] {};
\draw[double distance=1.5pt,scaling nfold=3,decoration={markings, mark=at position 0.65 with {\arrow[scale=0.7]{>}}}, postaction={decorate}] (0) to (1);
\draw (1)--(2)--(3);
\draw[line width=2,darkgray,decoration={markings, mark=at position 0.65 with {\arrow[scale=0.7]{<}}}, postaction={decorate}] (3) to node[midway,above,black] {\small{$\ell_1$}}(4a);
\draw[line width=2,darkgray,decoration={markings, mark=at position 0.65 with {\arrow[scale=0.7]{<}}}, postaction={decorate}] (3) to node[midway,below,black] {\small{$\ell_2$}}(4b);
\end{tikzpicture}  & \small{{$\!\begin{aligned}&\ell_{1},\ell_{2}\geq1\\&\textrm{gcd}\left(\ell_{1},\ell_{2}\right)=\textrm{gcd}\left(\ell_{1,2},3\right)=1\\&n\geq3\end{aligned}$}} \\ \cmidrule{2-4} 
 & $\mathcal{J}_{2,3}$ & \begin{tikzpicture}[baseline=0][baseline=0]
\tikzset{node distance=1cm};
\node (0) at (0,0) [gauge,label={[align=center]above:\small{$1$}}] {};
\node (1) at (1,0) [gauge,label={[align=center]above:\small{$2$}}] {};
\node (2) at (2,0) [gauge,label={[align=center]above:\small{$1$}}] {};
\draw[double distance=1.5pt,scaling nfold=2,decoration={markings, mark=at position 0.65 with {\arrow[scale=0.7]{>}}}, postaction={decorate}] (1) to (0);
\draw[double distance=1.5pt,scaling nfold=3,decoration={markings, mark=at position 0.65 with {\arrow[scale=0.7]{>}}}, postaction={decorate}] (1) to (2);
\end{tikzpicture} & - \\  \cmidrule{2-4}
 & $\mathcal{J}_{3,3}$ & \begin{tikzpicture}[baseline=0][baseline=0]
\tikzset{node distance=1cm};
\node (0) at (0,0) [gauge,label={[align=center]above:\small{$1$}}] {};
\node (1) at (1,0) [gauge,label={[align=center]above:\small{$2$}}] {};
\node (2) at (2,0) [gauge,label={[align=center]above:\small{$1$}}] {};
\draw[double distance=1.5pt,scaling nfold=3,decoration={markings, mark=at position 0.65 with {\arrow[scale=0.7]{>}}}, postaction={decorate}] (1) to (0);
\draw[double distance=1.5pt,scaling nfold=3,decoration={markings, mark=at position 0.65 with {\arrow[scale=0.7]{>}}}, postaction={decorate}] (1) to (2);
\end{tikzpicture} & -\\ \midrule[2pt]
\end{NiceTabular}
\caption{Classification of stable unitary quivers (see Definition \ref{def:quiver}). The gray-colored edges denote non-simply laced edges, defined by parameters $\ell$ and $k$. Condition $\prod_{i}k_{i}=\prod_{j}\ell_{j}$ for $\overline{h}_{n,\sigma}$ is equivalently described by the length function (see Definition \ref{def:quiver}) $L(1)=\prod^{n+1}_{j=2}\ell_{j}$, $L(2)=\prod^{n+1}_{j=2}k_{j}$, $L(i)=\prod_{j=i}^{n+1}k_{j}\prod_{m=2}^{i-1}\ell_{m}$ for the set of vertices $V=\left\{1,2,\dots,n+1\right\}$, starting the labelling on the upper-left vertex and continuing clock-wise. Quivers with subscript $n$ in their label/geometry consist of $\left(n+1\right)$-many vertices. Edges of the form $(k_i , \ell_i)$ with $k_i > 1$ and $\ell_i > 1$ are considered in Section \ref{sec:pq_edge}. In Sections \ref{sec:def} and \ref{sec:proofs}, one should impose $k_i = 1$ or $\ell_i = 1$. }
\label{tab:results}
\end{table}

%% file: fig_simply_laced.tex
\begin{figure}[t!]
\centering
        \begin{tikzpicture}[baseline=0]
        
        \tikzset{node distance=2cm};
        \node (1a) at (0,0) [gauge,label={[align=center]below:\small{$1$}}] {};
        \node (1b) [gauge,right of=1a,label={[align=center]below:\small{$1$}}] {};
        \draw[line width=2pt,darkgray] (1a)to node[midway, above, black] {\small{$N\geq2$}} (1b);
        \node (Label1) at (1,-0.5) {$A_{N-1}$};
                
        \node (2a) at (5,1.5) [gauge,label={[align=center]above:\small{$1$}}] {};
        \node (2b) [gauge,right of=2a,label={[align=center]above:\small{$1$}}] {};
        \node (2c) at (4,0) [gauge,label={[align=center]left:\small{$1$}}] {};
        \node (2d) at (8,0) [gauge,label={[align=center]right:\small{$1$}}] {};
        \node (2e) at (5,-1.5) [gauge,label={[align=center]below:\small{$1$}}] {};
        \node (2f) at (7,-1.5) [gauge,label={[align=center]below:\small{$1$}}] {};
        \node (DOT) at (6,-1.5) {$\cdots$};
        \node (End1) at (5.5,-1.5) {};
        \node (End2) at (6.5,-1.5) {};
        \node (Label2) at (6,-2) {$A^{(1)}_{n}$};
        \draw (2a)to (2b);
        \draw (2b)to (2d);
        \draw (2a)to (2c);
        \draw (2c)to (2e);
        \draw(2d)to (2f);
        \draw (2e)to (DOT);
        \draw (DOT)to (2f);

        \tikzset{node distance=1cm};
        \node (3a) at (9.2,0.8) [gauge,label={[align=center]above:\small{$1$}}] {};
        \node (3b) at (10,0) [gauge,label={[align=center]above:\small{$2$}}] {};
        \node (3h) at (9.2,-0.8) [gauge,label={[align=center]below:\small{$1$}}] {};
        \node (3c) [gauge,right of=3b,label={[align=center]above:\small{$2$}}] {};
        \node (DOT) [right of=3c] {$\cdots$};
        \node (3dot1) at (11.5,0) {};
        \node (3dot2) at (12,0) {};
        \node (3e) [gauge,right of=DOT,label={[align=center]above:\small{$2$}}]{};
        \node (3f) [gauge,right of=3e,label={[align=center]above:\small{$2$}}]{};
        \node (3g) at (14.8,0.8) [gauge,label={[align=center]above:\small{$1$}}]{};
        \node (3i) at (14.8,-0.8) [gauge,label={[align=center]below:\small{$1$}}]{};
        \node (Label3) at (11.75,-0.5) {$D^{(1)}_{n}$};
        \draw (3a)to (3b);
        \draw (3h)to (3b);
        \draw (3b)to (3c);
        \draw (3c)to (DOT);
        \draw (DOT)to (3e);
        \draw (3e)to (3f);
        \draw (3g)to (3f);
        \draw (3i)to (3f);    

        \tikzset{node distance=0.8cm};
        \node (4a) at (0,-3) [gauge,label={[align=center]above:\small{$1$}}] {};
        \node (4b) [gauge,right of=4a,label={[align=center]above:\small{$2$}}] {};
        \node (4c) [gauge,right of=4b,label={[align=center]above:\small{$3$}}] {};
        \node (4d) [gauge,right of=4c,label={[align=center]above:\small{$2$}}] {};
        \node (4e) [gauge,right of=4d,label={[align=center]above:\small{$1$}}] {};
        \node (4f) [gauge,below of=4c,label={[align=center]right:\small{$2$}}] {};
        \node (4g) [gauge,below of=4f,label={[align=center]right:\small{$1$}}] {};
        \node (Label4) at (0.8,-3.8) {$E^{(1)}_{6}$};
        \draw (4a)to (4b);
        \draw (4b)to (4c);
        \draw (4c)to (4d);
        \draw (4e)to (4d);
        \draw (4c)to (4f);
        \draw (4g)to (4f);

        \node (5a) at (4.2,-3) [gauge,label={[align=center]above:\small{$1$}}] {};
        \node (5b) [gauge,right of=5a,label={[align=center]above:\small{$2$}}] {};
        \node (5c) [gauge,right of=5b,label={[align=center]above:\small{$3$}}] {};
        \node (5d) [gauge,right of=5c,label={[align=center]above:\small{$4$}}] {};
        \node (5e) [gauge,right of=5d,label={[align=center]above:\small{$3$}}] {};
        \node (5f) [gauge,right of=5e,label={[align=center]above:\small{$2$}}] {};
        \node (5g) [gauge,right of=5f,label={[align=center]above:\small{$1$}}] {};
        \node (5h) [gauge,below of=5d,label={[align=center]right:\small{$2$}}] {};
        \node (Label5) at (5.8,-3.8) {$E^{(1)}_{7}$};
        \draw (5a)to (5b);
        \draw (5b)to (5c);
        \draw (5c)to (5d);
        \draw (5d)to (5e);
        \draw (5e)to (5f);
        \draw (5g)to (5f);
        \draw (5d)to (5h);

        \node (6a) at (10,-3) [gauge,label={[align=center]above:\small{$1$}}] {};
        \node (6b) [gauge,right of=6a,label={[align=center]above:\small{$2$}}] {};
        \node (6c) [gauge,right of=6b,label={[align=center]above:\small{$3$}}] {};
        \node (6d) [gauge,right of=6c,label={[align=center]above:\small{$4$}}] {};
        \node (6e) [gauge,right of=6d,label={[align=center]above:\small{$5$}}] {};
        \node (6f) [gauge,right of=6e,label={[align=center]above:\small{$6$}}] {};
        \node (6g) [gauge,right of=6f,label={[align=center]above:\small{$4$}}] {};
        \node (6h) [gauge,right of=6g,label={[align=center]above:\small{$2$}}] {};
        \node (6i) [gauge,below of=6f,label={[align=center]right:\small{$3$}}] {};
        \node (Label6) at (13.2,-3.8) {$E^{(1)}_{8}$};
        \draw (6a)to(6b);
        \draw (6b)to (6c);
        \draw (6c)to (6d);
        \draw (6d)to (6e);
        \draw (6e)to (6f);
        \draw (6f)to (6g);
        \draw (6g)to (6h);
        \draw (6f)to (6i);
        
        \end{tikzpicture}
\caption{All simply-laced quivers that are stable. There are the affine Dynkin quivers $A_n^{(1)}$ ($n \geq 2$), $D_n^{(1)}$ ($n \geq 4$), and $E_{6,7,8}^{(1)}$; conventions follow \cite{Fuchs:1997jv}. The subscript in the name of each quiver is one less than the number of nodes. In addition, there is the 2-vertex quiver $A_{N-1}$ with an edge of multiplicity $N$.}
\label{Fig:Simply_Laced}
\end{figure}
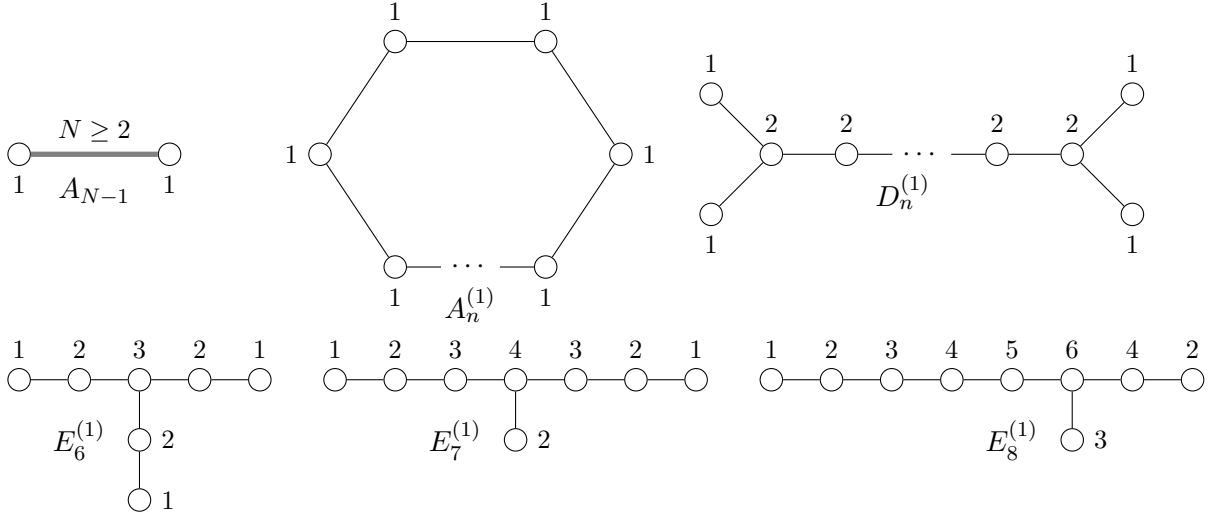

%% file: Dynkin_Table.tex
\begin{table}
\centering
\vspace*{-1.3cm}\begin{NiceTabular}{ccccc}
\midrule[2pt]
& Geometry & Quiver & Symmetry & Graph Name \\ \midrule[2pt]
\Block{11-1}<\rotate>{\large{\textsc{(Modified) Affine Dynkin-Type Quivers}}} 

& $\overline{\mathcal{O}_{\text{min}}}\left(\mathfrak{a}_{2n}\right)$ &         \begin{tikzpicture}[baseline=0]
            \tikzset{node distance=0.8cm};
            \node (1a) [gauge,label={[align=center]below:\small{$2$}}] {};
            \node (1b) [gauge,right of=1a,label={[align=center]below:\small{$2$}}] {};
            \node (1c) [gauge,right of=1b,label={[align=center]below:\small{$2$}}] {};
            \node (DOT) [right of=1c] {$\cdots$};
            \node (1dot) at (2.5,0) {};
            \node (2dot) at (3,0) {};
            \node (1e) [gauge,right of=DOT,label={[align=center]below:\small{$2$}}]{};
            \node (1f) [gauge,right of=1e,label={[align=center]below:\small{$2$}}]{};
            \node (1g) [gauge,right of=1f,label={[align=center]below:\small{$1$}}]{};
            \draw[double distance=1.5pt,scaling nfold=2,decoration={markings, mark=at position 0.65 with {\arrow[scale=0.7]{>}}}, postaction={decorate}] (1a)to (1b);
            \draw (1b)to (1c);
            \draw (1c) to (DOT);
            \draw (DOT) to (1e);
            \draw (1e)to (1f);
            \draw[double distance=1.5pt,scaling nfold=2,decoration={markings, mark=at position 0.65 with {\arrow[scale=0.7]{>}}}, postaction={decorate}] (1f)to (1g);
            \end{tikzpicture}
& $\mathfrak{a}_{2n}$ & $\tilde{B}_n^{(2)}$, $A_{2n}^{(2)}$ \\ \cmidrule{2-5} 
& $\overline{\mathcal{O}_{\text{min}}}\left(\mathfrak{a}_{2n-1}\right)$ &       \begin{tikzpicture}[baseline=0]
            \tikzset{node distance=0.8cm};
            \node (1a) at (-0.8,0.8) [gauge,label={[align=center]above:\small{$1$}}] {};
            \node (1b) at (0,0) [gauge,label={[align=center]above:\small{$2$}}] {};
            \node (1h) at (-0.8,-0.8) [gauge,label={[align=center]left:\small{$1$}}] {};
            \node (1c) [gauge,right of=1b,label={[align=center]above:\small{$2$}}] {};
            \node (DOT) [right of=1c] {$\cdots$};
            \node (1dot) at (1.5,0) {};
            \node (2dot) at (2,0) {};
            \node (1e) [gauge,right of=DOT,label={[align=center]above:\small{$2$}}]{};
            \node (1f) [gauge,right of=1e,label={[align=center]above:\small{$2$}}]{};
            \node (1g) [gauge,right of=1f,label={[align=center]above:\small{$2$}}]{};
            \draw[line width=2pt,darkgray, decoration={markings, mark=at position 0.65 with {\arrow[scale=0.7]{>}}}, postaction={decorate}] (1a)to node[midway,above,black] {\small{$\ell_1$}} (1b);
            \draw[line width=2pt,darkgray, decoration={markings, mark=at position 0.65 with {\arrow[scale=0.7]{>}}}, postaction={decorate}] (1h)to node[midway,right,black] {\small{$\ell_2$}} (1b);
            \draw (1b)to (1c);
            \draw (1c)to (DOT);
            \draw (DOT)to (1e);
            \draw (1e)to (1f);
            \draw[double distance=1.5pt,scaling nfold=2,decoration={markings, mark=at position 0.65 with {\arrow[scale=0.7]{<}}}, postaction={decorate}] (1f)to (1g);
            \end{tikzpicture}
 & $\mathfrak{a}_{2n-1}$ &  $C_n^{(2)}$, $A_{2n-1}^{(2)}$\\ \cmidrule{2-5} 
&$\overline{\mathcal{O}_{\text{min}}}\left(\mathfrak{b}_n\right)$ &         \begin{tikzpicture}[baseline=0]
            \tikzset{node distance=0.8cm};
            \node (1a) at (-0.8,0.8) [gauge,label={[align=center]above:\small{$1$}}] {};
            \node (1b) at (0,0) [gauge,label={[align=center]above:\small{$2$}}] {};
            \node (1h) at (-0.8,-0.8) [gauge,label={[align=center]left:\small{$1$}}] {};
            \node (1c) [gauge,right of=1b,label={[align=center]above:\small{$2$}}] {};
            \node (DOT) [right of=1c] {$\cdots$};
            \node (1dot) at (1.5,0) {};
            \node (2dot) at (2,0) {};
            \node (1e) [gauge,right of=DOT,label={[align=center]above:\small{$2$}}]{};
            \node (1f) [gauge,right of=1e,label={[align=center]above:\small{$2$}}]{};
            \node (1g) [gauge,right of=1f,label={[align=center]above:\small{$1$}}]{};
            \draw[line width=2pt,darkgray, decoration={markings, mark=at position 0.65 with {\arrow[scale=0.7]{>}}}, postaction={decorate}] (1a)to node[midway,above,black] {\small{$\ell_1$}} (1b);
            \draw[line width=2pt,darkgray, decoration={markings, mark=at position 0.65 with {\arrow[scale=0.7]{>}}}, postaction={decorate}] (1h)to node[midway,right,black] {\small{$\ell_2$}} (1b);
            \draw (1b)to (1c);
            \draw (1c)to (DOT);
            \draw (DOT)to (1e);
            \draw (1e)to (1f);
            \draw[double distance=1.5pt,scaling nfold=2,decoration={markings, mark=at position 0.65 with {\arrow[scale=0.7]{>}}}, postaction={decorate}] (1f)to (1g);
            \end{tikzpicture}
& $\mathfrak{b}_{n}$ &  $B^{(1)}_n$\\ \cmidrule{2-5} 
&$\overline{\mathcal{O}_{\text{min}}}\left(\mathfrak{d}_n\right)$ &             \begin{tikzpicture}[baseline=0]
                \tikzset{node distance=0.8cm};
                \node (1a) at (-0.8,0.8) [gauge,label={[align=center]above:\small{$1$}}] {};
                \node (1b) at (0,0) [gauge,label={[align=center]above:\small{$2$}}] {};
                \node (1h) at (-0.8,-0.8) [gauge,label={[align=center]below:\small{$1$}}] {};
                \node (1c) [gauge,right of=1b,label={[align=center]above:\small{$2$}}] {};
                \node (DOT) [right of=1c] {$\cdots$};
                \node (1dot) at (1.5,0) {};
                \node (2dot) at (2,0) {};
                \node (1e) [gauge,right of=DOT,label={[align=center]above:\small{$2$}}]{};
                \node (1f) [gauge,right of=1e,label={[align=center]above:\small{$2$}}]{};
                \node (1g) at (4.3,0.8) [gauge,label={[align=center]above:\small{$1$}}]{};
                \node (1i) at (4.3,-0.8) [gauge,label={[align=center]below:\small{$1$}}]{};
                \draw[line width=2pt,darkgray, decoration={markings, mark=at position 0.65 with {\arrow[scale=0.7]{>}}}, postaction={decorate}] (1a)to node[midway,above,black] {\small{$\ell_1$}} (1b);
                \draw[line width=2pt,darkgray, decoration={markings, mark=at position 0.65 with {\arrow[scale=0.7]{>}}}, postaction={decorate}] (1h)to node[midway,right,black] {\small{$\ell_2$}} (1b);
                \draw (1b)to (1c);
                \draw (1c)to (DOT);
                \draw (DOT)to (1e);
                \draw (1e)to (1f);
                \draw[line width=2pt,darkgray, decoration={markings, mark=at position 0.65 with {\arrow[scale=0.7]{>}}}, postaction={decorate}] (1g)to node[midway,above,black] {\small{$\ell_3$}} (1f);
                \draw[line width=2pt,darkgray, decoration={markings, mark=at position 0.65 with {\arrow[scale=0.7]{>}}}, postaction={decorate}] (1i)to node[midway,left,black] {\small{$\ell_4$}} (1f);
                \end{tikzpicture}
& $\mathfrak{d}_{n}$ &  $D^{(1)}_n$ \\ \cmidrule{2-5}  
& $\overline{\mathcal{O}_{\text{min}}}\left(\mathfrak{d}_{n+1}\right)$ &         \begin{tikzpicture}[baseline=0]
            \tikzset{node distance=0.8cm};
            \node (1a) [gauge,label={[align=center]above:\small{$1$}}] {};
            \node (1b) [gauge,right of=1a,label={[align=center]above:\small{$2$}}] {};
            \node (1c) [gauge,right of=1b,label={[align=center]above:\small{$2$}}] {};
            \node (DOT) [right of=1c] {$\cdots$};
            \node (1dot) at (2.5,0) {};
            \node (2dot) at (3,0) {};
            \node (1e) [gauge,right of=DOT,label={[align=center]above:\small{$2$}}]{};
            \node (1f) [gauge,right of=1e,label={[align=center]above:\small{$2$}}]{};
            \node (1g) [gauge,right of=1f,label={[align=center]above:\small{$1$}}]{};
            \draw[double distance=1.5pt,scaling nfold=2,decoration={markings, mark=at position 0.65 with {\arrow[scale=0.7]{>}}}, postaction={decorate}] (1b)to (1a);
            \draw (1b)to (1c);
            \draw (1c)to (DOT);
            \draw (DOT)to (1e);
            \draw (1e)to (1f);
            \draw[double distance=1.5pt,scaling nfold=2,decoration={markings, mark=at position 0.65 with {\arrow[scale=0.7]{>}}}, postaction={decorate}] (1f)to (1g);
            \end{tikzpicture}
& $\mathfrak{d}_{n+1}$ & $B_n^{(2)}$, $D_{n+1}^{(2)}$ \\ \cmidrule{2-5} 
& $\overline{\mathcal{O}_{\text{min}}}\left(\mathfrak{d}_{4}\right)$ &         \begin{tikzpicture}[baseline=0]
            \tikzset{node distance=1cm};
            \node (1a) [gauge,label={[align=center]above:\small{$1$}}] {};
            \node (1b) [gauge,right of=1a,label={[align=center]above:\small{$2$}}] {};
            \node (1c) [gauge,right of=1b,label={[align=center]above:\small{$3$}}] {};
            \draw[line width=2pt,darkgray, decoration={markings, mark=at position 0.65 with {\arrow[scale=0.7]{>}}}, postaction={decorate}] (1a)to node[midway,above,black] {\small{$\ell_1$}} (1b);
            \draw[double distance=1.5pt,scaling nfold=3, decoration={markings, mark=at position 0.65 with {\arrow[scale=0.7]{>}}}, postaction={decorate}] (1c)to (1b);
            \end{tikzpicture}
& $\mathfrak{d}_{4}$ & $G_2^{(3)}$, $D_4^{(3)}$\\ \cmidrule{2-5} 
& $\overline{\mathcal{O}_{\text{min}}}\left(\mathfrak{e}_6\right)$ &       \begin{tikzpicture}[baseline=0]
                \tikzset{node distance=1cm};
                \node (1a) at (0,0) [gauge,label={[align=center]above:\small{$1$}}] {};
                \node (1b) [gauge,right of=1a,label={[align=center]above:\small{$2$}}] {};
                \node (1c) [gauge,right of=1b,label={[align=center]above:\small{$3$}}] {};
                \node (1d) [gauge,right of=1c,label={[align=center]above:\small{$2$}}] {};
                \node (1e) [gauge,right of=1d,label={[align=center]above:\small{$1$}}] {};
                \node (1f) [gauge,below of=1c,label={[align=center]right:\small{$2$}}] {};
                \node (1g) [gauge,below of=1f,label={[align=center]right:\small{$1$}}] {};
                \draw[line width=2pt,darkgray, decoration={markings, mark=at position 0.65 with {\arrow[scale=0.7]{>}}}, postaction={decorate}] (1a)to node[midway,above,black] {\small{$\ell_1$}} (1b);
                \draw (1b)to (1c);
                \draw (1c)to (1d);
                \draw[line width=2pt,darkgray, decoration={markings, mark=at position 0.65 with {\arrow[scale=0.7]{>}}}, postaction={decorate}] (1e)to node[midway,above,black] {\small{$\ell_2$}} (1d);
                \draw (1c)to (1f);
                \draw[line width=2pt,darkgray, decoration={markings, mark=at position 0.65 with {\arrow[scale=0.7]{>}}}, postaction={decorate}] (1g)to node[midway,right,black] {\small{$\ell_3$}} (1f);
                \end{tikzpicture}
 & $\mathfrak{e}_6$ & $E^{(1)}_6$\\ \cmidrule{2-5} 
& $\overline{\mathcal{O}_{\text{min}}}\left(\mathfrak{e}_{6}\right)$ &             \begin{tikzpicture}[baseline=0]
                \tikzset{node distance=1cm};
                \node (1a) at (0,0) [gauge,label={[align=center]above:\small{$2$}}] {};
                \node (1b) [gauge,right of=1a,label={[align=center]above:\small{$4$}}] {};
                \node (1c) [gauge,right of=1b,label={[align=center]above:\small{$3$}}] {};
                \node (1d) [gauge,right of=1c,label={[align=center]above:\small{$2$}}] {};
                \node (1e) [gauge,right of=1d,label={[align=center]above:\small{$1$}}] {};
                \draw (1a)to (1b);
                \draw[double distance=1.5pt,scaling nfold=2,decoration={markings, mark=at position 0.65 with {\arrow[scale=0.7]{>}}}, postaction={decorate}] (1b)to (1c);
                \draw (1c)to (1d);
                \draw[line width=2pt,darkgray, decoration={markings, mark=at position 0.65 with {\arrow[scale=0.7]{>}}}, postaction={decorate}] (1e)to node[midway,above,black] {\small{$\ell_1$}} (1d);
                \end{tikzpicture}
& $\mathfrak{e}_{6}$ & $F_4^{(2)}$, $E_6^{(2)}$\\ \cmidrule{2-5} 
&$\overline{\mathcal{O}_{\text{min}}}\left(\mathfrak{e}_7\right)$&               
\begin{tikzpicture}[baseline=0]
                \tikzset{node distance=0.8cm};
                \node (1a) at (0,0) [gauge,label={[align=center]above:\small{$1$}}] {};
                \node (1b) [gauge,right of=1a,label={[align=center]above:\small{$2$}}] {};
                \node (1c) [gauge,right of=1b,label={[align=center]above:\small{$3$}}] {};
                \node (1d) [gauge,right of=1c,label={[align=center]above:\small{$4$}}] {};
                \node (1e) [gauge,right of=1d,label={[align=center]above:\small{$3$}}] {};
                \node (1f) [gauge,right of=1e,label={[align=center]above:\small{$2$}}] {};
                \node (1g) [gauge,right of=1f,label={[align=center]above:\small{$1$}}] {};
                \node (1h) [gauge,below of=1d,label={[align=center]right:\small{$2$}}] {};
                \draw[line width=2pt,darkgray, decoration={markings, mark=at position 0.65 with {\arrow[scale=0.7]{>}}}, postaction={decorate}] (1a)to node[midway,above,black] {\small{$\ell_1$}} (1b);
                \draw (1b)to (1c);
                \draw (1c)to (1d);
                \draw (1d)to (1e);
                \draw (1e)to (1f);
                \draw[line width=2pt,darkgray, decoration={markings, mark=at position 0.65 with {\arrow[scale=0.7]{>}}}, postaction={decorate}] (1g)to node[midway,above,black] {\small{$\ell_2$}} (1f);
                \draw (1d)to (1h);
                \end{tikzpicture}
& $\mathfrak{e}_7$ &  $E^{(1)}_7$ \\ \cmidrule{2-5} 
& $\overline{\mathcal{O}_{\text{min}}}\left(\mathfrak{e}_8\right)$ &             
\begin{tikzpicture}[baseline=0]
                \tikzset{node distance=0.8cm};
                \node (1a) [gauge,label={[align=center]above:\small{$1$}}] {};
                \node (1b) [gauge,right of=1a,label={[align=center]above:\small{$2$}}] {};
                \node (1c) [gauge,right of=1b,label={[align=center]above:\small{$3$}}] {};
                \node (1d) [gauge,right of=1c,label={[align=center]above:\small{$4$}}] {};
                \node (1e) [gauge,right of=1d,label={[align=center]above:\small{$5$}}] {};
                \node (1f) [gauge,right of=1e,label={[align=center]above:\small{$6$}}] {};
                \node (1g) [gauge,right of=1f,label={[align=center]above:\small{$4$}}] {};
                \node (1h) [gauge,right of=1g,label={[align=center]above:\small{$2$}}] {};
                \node (1i) [gauge,below of=1f,label={[align=center]right:\small{$3$}}] {};
                \draw[line width=2pt,darkgray, decoration={markings, mark=at position 0.65 with {\arrow[scale=0.7]{>}}}, postaction={decorate}] (1a)to node[midway,above,black] {\small{$\ell$}} (1b);
                \draw (1b)to (1c);
                \draw (1c)to (1d);
                \draw (1d)to (1e);
                \draw (1e)to (1f);
                \draw (1f)to (1g);
                \draw (1g)to (1h);
                \draw (1f)to (1i);
                \end{tikzpicture}
& $\mathfrak{e}_8$ &  $E^{(1)}_8$ \\ \cmidrule{2-5} 
& $\overline{\mathcal{O}_{\text{min}}}\left(\mathfrak{f}_4\right)$ &         
\begin{tikzpicture}[baseline=0]
                \tikzset{node distance=1cm};
                \node (1a) at (0,0) [gauge,label={[align=center]above:\small{$1$}}] {};
                \node (1b) [gauge,right of=1a,label={[align=center]above:\small{$2$}}] {};
                \node (1c) [gauge,right of=1b,label={[align=center]above:\small{$3$}}] {};
                \node (1d) [gauge,right of=1c,label={[align=center]above:\small{$2$}}] {};
                \node (1e) [gauge,right of=1d,label={[align=center]above:\small{$1$}}] {};
                \draw[line width=2pt,darkgray, decoration={markings, mark=at position 0.65 with {\arrow[scale=0.7]{>}}}, postaction={decorate}] (1a)to node[midway,above,black] {\small{$\ell_1$}} (1b);
                \draw (1b)to (1c);
                \draw[double distance=1.5pt,scaling nfold=2,decoration={markings, mark=at position 0.65 with {\arrow[scale=0.7]{>}}}, postaction={decorate}] (1c)to (1d);
                \draw[line width=2pt,darkgray, decoration={markings, mark=at position 0.65 with {\arrow[scale=0.7]{>}}}, postaction={decorate}] (1e)to node[midway,above,black] {\small{$\ell_2$}} (1d);
                \end{tikzpicture}
& $\mathfrak{f}_4$ & $F^{(1)}_4$\\ \cmidrule{2-5} 
&$\overline{\mathcal{O}_{\text{min}}}\left(\mathfrak{g}_2\right)$ &      \begin{tikzpicture}[baseline=0]
            \tikzset{node distance=1cm};
            \node (1a) [gauge,label={[align=center]above:\small{$1$}}] {};
            \node (1b) [gauge,right of=1a,label={[align=center]above:\small{$2$}}] {};
            \node (1c) [gauge,right of=1b,label={[align=center]above:\small{$1$}}] {};
            \draw[line width=2pt,darkgray, decoration={markings, mark=at position 0.65 with {\arrow[scale=0.7]{>}}}, postaction={decorate}] (1a)to node[midway,above,black] {\small{$\ell_1$}} (1b);
            \draw[double distance=1.5pt,scaling nfold=3, decoration={markings, mark=at position 0.65 with {\arrow[scale=0.7]{<}}}, postaction={decorate}] (1c)to (1b);
            \end{tikzpicture}
& $\mathfrak{g}_2$ &  $G^{(1)}_2$\\

\midrule[2pt]
\end{NiceTabular}
\caption{All affine Dynkin-type quivers that allow for a modification by a non-simply laced edge connected to a vertex of weight $1$, see Section~\ref{subsec:trivial_edges}. In any quiver, all $\ell_i$ should be pairwise coprime. For $\overline{\mathcal{O}_{\text{min}}}\left(\mathfrak{a}_{2n-1}\right)$ and $\overline{\mathcal{O}_{\text{min}}}\left(\mathfrak{e}_{6}\right)$, the $\ell_i$ are odd. For $\overline{\mathcal{O}_{\text{min}}}\left(\mathfrak{f}_{4}\right)$, $\ell_2$ is odd. For $\overline{\mathcal{O}_{\text{min}}}\left(\mathfrak{d}_{4}\right)$, $\mathrm{gcd}(\ell_1,3)=1$.  The naming of the quivers follows the conventions of affine Dynkin diagram of \cite[Tab.~VIII]{Fuchs:1997jv}. All the quivers in the first five rows have $n+1$ vertices.}
\label{tab:dynkin}
\end{table}